\documentclass[12pt,reqno]{amsart}

\usepackage[active]{srcltx}

\usepackage{latexsym}
\usepackage{graphicx}
\usepackage{hyperref} 
\hypersetup{
    colorlinks,
    citecolor=black,
    filecolor=black,
    linkcolor=black,
    urlcolor=black
}

\newcommand{\Sch}{Schr\"odinger }

\usepackage{color}
\usepackage{color}

\newcommand{\baa}{\begin{align*}}
\newcommand{\eaa}{\end{align*}}
\newcommand{\bea}{\begin{eqnarray*} }
\newcommand{\eea}{\end{eqnarray*} }
\newcommand{\beq}{\begin{equation} }
\newcommand{\eeq}{\end{equation} }
\newcommand{\bp}{\begin{proposition}}
\newcommand{\ep}{\end{proposition}}
\newcommand{\bt}{\begin{theorem}}
\newcommand{\et}{\end{theorem}}
\newcommand{\bpf}{\begin{proof}}
\newcommand{\epf}{\end{proof}}
\newcommand{\bl}{\begin{lemma}}
\newcommand{\el}{\end{lemma}}
\newcommand{\bc}{\begin{cor}}
\newcommand{\ec}{\end{cor}}
\newcommand{\bd}{\begin{definition}}
\newcommand{\ed}{\end{definition}}

\newcommand{\ihbar}{{\frac{i}{h}}}

\def\del{\partial}

\newcommand{\twiddle}[1]{\ensuremath{\widetilde{#1}}}

\newcommand{\be}{\begin{equation} }
\newcommand{\ee}{\end{equation} }
\newcommand{\bee}{\begin{eqnarray} }
\newcommand{\eee}{\end{eqnarray} }

\newcommand{\gives}{\ensuremath{\rightarrow}}

\newcommand{\EE}[1]{\mathbb E}
\newcommand{\abs}[1]{\left\lvert #1 \right\rvert}
\newcommand{\norm}[1]{\left\lVert#1\right\rVert}
\newcommand{\lr}[1]{\ensuremath{\left(#1\right)}}
\renewcommand{\Re}{\ensuremath{\mathrm{Re} \ }}
\renewcommand{\Im}{\ensuremath{\mathrm{Im} \ }}
\newcommand{\dell}{\ensuremath{\partial}}

\newcommand{\twomat}[4]{\ensuremath{ \left(\begin{array}{cc} #1 & #2 \\
#3 & #4 \end{array}\right)}}

\renewcommand{\Re}{{\operatorname{Re}\,}}
\renewcommand{\Im}{{\operatorname{Im}\,}}

\renewcommand{\epsilon}{\varepsilon}

\newcommand{\wt}{\widetilde}

\newcommand{\R}{{\mathbb R}}

\newcommand{\half}{{\textstyle \frac 12}}

\renewcommand{\phi}{\varphi}

\newcommand{\scal}{\mathcal{S}}

\newcommand{\ucal}{\mathcal{U}}

\newcommand{\wcal}{\mathcal{W}}

\DeclareMathOperator{\Ai}{Ai}

\newtheorem{theo}{{\sc Theorem}}[section]
\newtheorem{lem}[theo]{{\sc Lemma}}
\newtheorem{prop}[theo]{{\sc Proposition}}

\newtheorem{defin}[theo]{{\sc Definition}}
\newtheorem{cor}[theo]{{\sc Corollary}}

\newtheorem{Thm}[theo]{{\sc Theorem}}

\newtheorem{proposition}[theo]{{\sc Proposition}}
\newtheorem{lemma}[theo]{{\sc Lemma}}

\newenvironment{rem}{\medskip\noindent{\it Remark:\/} }{\medskip}

\DeclareMathOperator{\Id}{Id}

\begin{document}

\title{Scaling asymptotics of spectral Wigner functions }

\address{Department of Mathematics, Texas A\& M University, College Station, TX
77845, USA}
\email[B. Hanin]{bhanin@princeton.edu}
\address{Department of Operations Research and Financial Engineering, Princeton  University, Princeton, NJ, 
08540, USA}
\email[S. Zelditch]{zelditch@math.northwestern.edu}

\author{Boris Hanin, Steve Zelditch  }
\thanks{S.Z. was partially supported by NSF grant DMS-1810747.} \thanks{B.H. was funded by NSF CAREER grant DMS-2143754
as well as NSF grants DMS-1855684, DMS-2133806 and an ONR MURI on Foundations of Deep Learning.}

\maketitle

\begin{abstract} We prove that smooth Wigner-Weyl spectral sums at an energy level $E$ exhibit Airy
scaling asymptotics across the classical energy surface $\Sigma_E$. This was proved earlier by the authors
for the isotropic harmonic oscillator and the proof is extended in this article to all quantum Hamiltonians
$-\hbar^2 \Delta + V$ where $V$ is a confining potential with at most quadratic growth at infinity. The main 
tools are the Herman-Kluk initial value  parametrix for the propagator and the Chester-Friedman-Ursell normal form for complex phases with
a one-dimensional cubic degeneracy. This gives a rigorous account of Airy scaling asymptotics of  spectral Wigner 
distributions of M.V. Berry, A. Ozorio de Almeida and other physicists.

\end{abstract}


This article is thus concerned with  spectral  Wigner functions for 
 a Schr\"{o}dinger operator
\begin{equation} \label{Hh} \widehat{H}_\hbar = -\frac{\hbar^2}{2}\Delta_{\R^d} +
V(x),\end{equation}
where $V\in C^\infty(\R^d)$ is a real-valued potential that satisfies $V(x) \to \infty$ as $||x|| \to \infty$, and is at most of quadratic growth at infinity,
\begin{equation} \label{SUBQ} \forall \gamma\in \mathbb N^{d}\text{ s.t. }\abs{\gamma}\geq 2~~~\exists C_\gamma>0,\,\, \sup_x \abs{\dell_x^\gamma V(x)}<C_\gamma.\end{equation}
Let $\{\psi_j(\hbar)\}$ be a complete orthonormal basis for $L^2(\R^d,dx)$ consisting of  eigenfunctions of $\widehat{H}_{\hbar}$:
$$\widehat{H}_{\hbar} \psi_j (\hbar, x) = E_j(\hbar)  \psi_j(\hbar, x),$$  
and let \begin{equation} W_{\psi_j, \psi_j}(x,\xi) =
(2 \pi\hbar)^{-d} \int_{\R^d}  \psi_j(\hbar, x+\frac{v}{2})  \psi_j(\hbar, x-\frac{v}{2} ) e^{-\frac{i}{\hbar} v \xi} \frac{dv}{(2\pi h)^d}. \label{E:Wignerefn}
\end{equation} 
  be the  Wigner functions of the individual  eigenstates. The Weyl-Wigner spectral functions for an energy interval $I_{\hbar} (E)= [E - a, E + b]$ centered 
  at $E$  are defined by,
\begin{equation} \label{WEYLWIGNER} W_{\hbar, I_{\hbar} } (x, \xi): = 
  \sum_{j: E_j(\hbar) \in I_{\hbar} (E) } W_{\psi_j, \psi_j}(x,\xi). \end{equation}
In the recent articles \cite{HZ20, HZb}, the authors obtained scaling asymptotics of Airy type for \eqref{WEYLWIGNER} around the energy 
surface \begin{equation} \label{SIGMAE} \Sigma_E = \{(x, \xi): H(x, \xi) = E\}, \;\; H(x, \xi) = ||\xi||^2 + V(x) \end{equation}
in the special case of the isotropic Harmonic oscillator where $V(x) = ||x||^2$. The asymptotics involve
two types of localization: (i) spectral localization to the interval $I_{\hbar}(E)$ and (ii) phase space localization,  where
$(x, \xi)$ is localized either near $\Sigma_E$ or at some prescribed distance from it, either in the energy ball, 
$$B_E  = \{(x, \xi): H(x, \xi) \leq E \}$$
or outside the energy surface. 
The purpose of this article is to present one generalization of
the Airy scaling results to all Schr\"odinger operators \eqref{Hh} for which $H$ is strictly convex, hence
\begin{equation} \label{HYP} B_E \; \rm{is\; strictly\; convex} \;  \end{equation} and to relate the scaling results to statements in M.V. Berry's article
\cite{Ber89} as well as to the related article \cite{O98} of Ozorio de Almeida (see also \cite{TL}).

The result we generalize here is the following \cite[Theorem 1.4]{HZ20}. Let $E_N(\hbar) = \hbar (N + \frac{d}{2})$ be the $N$th {\it distinct}
 eigenvalue of the isotropic harmonic oscillator (see Section \ref{ISOHO} for notation and background) and define  $\hbar = \hbar_N(E):=\frac{E}{N+\frac{d}{2}},$  so that   $E_N(\hbar) = E$. Let  $\Pi_{\hbar, E_N(\hbar)}$ be the orthogonal projection onto the corresponding eigenspace $V_N(\hbar)$ and let $ W_{\hbar, E_N(\hbar)}(x, \xi)$ be its Wigner distribution (in the notation of \eqref{WEYLWIGNER} this means we choose $a=b=0$).
  Then, with $s =
\frac{H(x, \xi)}{E}$,  and  where  for simplicity we assume that  $s =  \frac{H(x, \xi)}{E} \in (0,1]$, it follows from \cite[(50)]{HZ20} that, 
\begin{align}
\label{WAIRY} W_{\hbar, E_N(\hbar)}(x, \xi)&=\frac{2}{(2\pi \hbar)^d} \left[\hbar_E^{1/3}\Ai\lr{\hbar_E^{-2/3}B^2\lr{s}}\lr{1+O\lr{(s-1)^{2/3}}}+\epsilon_1(h_E,s)\right]
\end{align}
where we've sest
\begin{equation}\label{BFORMULA} 
B(s) = i (3 \beta(s)/2)^{1/3}, \qquad 
\beta(s) := \half [\cos^{-1} \sqrt{s} - \sqrt{s - s^2}].
\end{equation} 
Note that in  \cite[(50)]{HZ20} we use the
large parameter $\nu = \frac{2}{\hbar_E}$  and that the argument of the Airy function there is $\nu^{\frac{2}{3}} B^2(s)$; this introduces
some additional universal constants in the asymptotics. Note also that since we've restricted for simplicity to $s\leq 1$, the formula \eqref{BFORMULA} for $B(s)$ is pure imaginary, so that its square is negative and the Airy function is evaluated in its oscillatory region.
In Section \ref{ISOHO} we interpret the formulae geometrically and simplify \eqref{WAIRY} when $(x, \xi)$ is 
very close to $\Sigma_E$.

 The present article generalizes the scaling results \eqref{WAIRY}  for the more general Hamiltonians \eqref{Hh}. There are both sharp and
 smooth versions of the scaling results, depending on whether we sum over an energy interval as in \eqref{WEYLWIGNER} or
 whether we use a smooth cutoff function $f$ to define  general smoothed Weyl-Wigner sums,
\begin{equation} \label{Whfdelta} \begin{array}{lll}
W_{\hbar, f, \delta,E}(x, \xi): =  \sum_{j=1}^{\infty}  f\lr{\hbar^{-\delta}(E- E_j(\hbar))}W_{ \psi_j, \psi_j}(x, \xi)  \end{array}
\end{equation} 
where $f \in \scal(\R)$ is a Schwartz function with $\hat{f} \in C_0^{\infty}(\R)$. In this article we only consider $\delta =1$ and we concentrate on the behavior of \eqref{Whfdelta} in a thin $\hbar^{1/3}$ neighborhood of $\Sigma_E$.
In 
Section \ref{OTHER} we review other scaling results in \cite{HZ20,HZb} where  $\delta =0,  \frac{2}{3}$, i.e. where we sum over
larger spectral intervals.

To state our first result about $W_{\hbar, f, 1, E}(x, \xi)$, we need some notation. Let us denote the Hamilton vector field of $H : T^*\R^d \to \R$ by $\Xi_H$, and the  Hamiltonian flow of $\Xi_H$ by $\Phi^t$. The classical trajectory  with initial data $(q,p)$ is denoted  by
\begin{equation} \label{PHIt} \Phi^t(q, p)  =(q_t,p_t),
\end{equation}
and we will write
\begin{equation}\label{E:Mt-def-intro}
    M_t(q,p):=D_{q,p}\Phi^t(q,p)
\end{equation}
for its Jacobian. It was pointed out in \cite{Ber89, O98} that the asymptotics of $W_{\hbar, f, 1, E}(x, \xi)$ for $(x, \xi) \in B_E$ sufficiently close to $\Sigma_E$ is governed in a sense that we will describe by the midpoint map,  $\Theta^*:\R \times \Sigma_E\gives B_E$ given by 
\begin{equation} \label{PSItDEF}\Theta^*(t, q, p) =  \Theta^t(q, p): = \half (I + \Phi^t)(q,p) = \left(\frac{q+ q_t}{2}, \frac{p+p_t}{2}\right) . \end{equation}
The assumption \eqref{HYP} ensures that \eqref{PSItDEF} takes its values in $B_E$. The asymptotics of $W_{\hbar,f,1,E}(x,\xi)$ depend on the solutions  $(t, q, p)$ of the equations
\begin{equation} \label{SOLVE} \Theta^t(q, p) = (x, \xi), \;\; (q, p) \in \Sigma_E, \;\; (x, \xi) \in B_E, \end{equation}
with $(E, x, \xi)$ are fixed, which we will see in  Proposition \ref{PROPPROPt} below arise as critical point equations for a certain phase function related to $W_{\hbar,f,1,E}$. An important point is that there is an obvious symmetry in the equation \eqref{SOLVE}, namely \begin{equation} \label{SYM} 
\Theta^t(q, p) = \Theta^{-t} (\Phi^t(q, p)), \end{equation}
which fixes the Hamiltonian arc from $(q, p)$ to $(q_t, p_t)$ but reverses the endpoints. It is proved in Section \ref{EFOLDSECT} (see Lemma \ref{IFTLEM}) that under assumption \ref{HYP}, there exists $\epsilon > 0$ so that for $(x, \xi)$ in an
$\epsilon$-tube around $\Sigma_E$ there exist  unique $(t, q,p)$ up to the symmetry \eqref{SYM} satisfying the midpoint equation,
\begin{equation} \label{MIDPOINT}
\Theta(t, q, p) = (x, \xi). 
\end{equation}
Put another way, let us denote by  $ s\in [0,1] \to \gamma_{t, x, \xi}(s) $ the unique Hamiltonian arc (up to the involution \eqref{SYM}) with $|t|$ sufficiently small on $\Sigma_E$  such that 
\begin{equation} \label{gammaDEF} (x, \xi) = \half (\gamma_{t, x, \xi}(0) + \gamma_{t, x, \xi}(t)).
\end{equation}
 In Lemma \ref{ODDLEM} 
it is proved that, there exists $\epsilon_0 > 0$ so that if  $(x, \xi)$ is $\epsilon_0$-close to $\Sigma_E$, then there exist exactly two critical times $t = t_{\pm}(E, x, \xi) \in (-\epsilon_0, \epsilon_0)$ (see Proposition \ref{PROPPROP})  of the form $t_{\pm} = \pm t$ such that $\Theta^t(q, p) = (x, \xi), H(q, p) = E. $ 
For ease of future reference, we make the following

\begin{defin} \label{tDEF} We denote by  $t=t_{\pm} (E,x,\xi)$ with $t_{-} (E, x, \xi) = - t_+(E, x, \xi)$ the critical times where $\Theta^t(q, p) = (x, \xi), H(q, p) = E. $ 
\end{defin}

 A direct computation provided in \eqref{dtE} shows that $\frac{dt}{dE}(E,x,\xi)$ is well-defined for each $E>0$ and $(x,\xi)$ in a sufficiently small tubular neighborhood of $\Sigma_E$. With this notation, we are now ready to present our first result.

\begin{Thm}\label{T:SMOOTH-Prop}
Let $\widehat{H}_\hbar$ be a Hamiltonian as in \eqref{Hh} satisfying the assumptions \eqref{SUBQ} and \eqref{HYP}. Further, fix $f \in \scal(\R)$ such that
\[
\hat{f} \in C_0^{\infty},\qquad  \mathrm{Supp} \hat{f}  \subset (-a, a),\qquad \hat{f}(\tau) = 1\,\,\forall  \tau \in (-a/2, a/2)
\]
where $a < T_{\min}$ is less than the periodic of the minimal periodic orbit on $\Sigma_E$. Then  for $(x, \xi) \in B_E$ such that
$H(x, \xi) = E + O(\hbar^{2/3}), $ there exist $b, \sigma, u_{0, \nu}, u_{1, \nu}  \in C^{\infty}(B_E), $ such that the smoothed spectral Weyl-Wigner functions $W_{\hbar,f,1,E}$ defined in \eqref{Whfdelta}  with $\delta =1$ admit the pointwise semi-classical asymptotic expansion
\begin{equation} \label{AIRY} \begin{array}{lll} W_{\hbar, f, 1,E}(x, \xi) & \simeq &     \hbar^{-d} 
\left( \hbar^{\frac{1}{3}} \mathrm{Ai}(- \hbar^{-\frac{2}{3}} \rho(x, \xi) ) \sum_{\nu=0}^{\infty} u_{0, \nu}(x, \xi) \hbar^{\nu} \right)+ \\&&\\
& + & \hbar^{-d} \left( \hbar^{\frac{2}{3}} \mathrm{Ai}'(-\hbar^{-\frac{2}{3}} \rho(x, \xi)) \sum_{\nu=0}^{\infty} u_{1, \nu}(x, \xi) \hbar^{\nu} \right), \end{array} \end{equation}
where  
\begin{equation} \label{rhoform} \begin{array}{l} \frac{4}{3} \rho^{3/2}=  \half \int_{t_{-}(E, x, \xi)}^{t_{+}(E, x, \xi)} (p - \xi) \cdot dq =    \int_{\beta_{t, x, \xi} } \omega, \end{array} \end{equation} 
where $\beta$ is defined in the next paragraph 
and where the leading coefficient $u_{0,0}$ is given by,
\begin{equation}\label{MATCHa}  u_{0,0} = \sqrt{\pi} \rho(x,\xi)^{1/4} \left| \frac{dt_+}{dE}(E,x,\xi) \det (1 + M_{t_+(E,x,\xi)}(x, \xi))\right|^{-\half},  \end{equation}

\end{Thm}

The integral in \eqref{rhoform} is the integral over the oriented closed curve obtained by connecting the endpoints of $\gamma_{t, x, \xi}$ by the  chord
 \[
 \alpha_{t, x, \xi}(s) = (1 -s) \gamma_{t, x, \xi}(0)  + s \gamma_{t, x, \xi},\qquad s\in [0,1].
 \]
 Since the closed curve bounds the two-dimensional surface
   $\beta_{t, x, \xi}$ consisting 
  of line segments joining $(x, \xi)$
to points of the Hamilton orbit, the integral \eqref{rhoform} equals the oriented area
  $ \int_{\beta} \omega $
  where  $\omega = d p \wedge dq$ is the standard symplectic form of $T^*\R^d$.  In Section \ref{ISOHO} we show
  that $\rho$ in \eqref{AIRY} equals $-B^2(s)$ up to universal constants in \eqref{WAIRY} in the case of the isotropic harmonic oscillator.
  
The  sign conventions in the argument of the Airy function is discussed below Proposition \ref{HORMPROP}. It is clearly
consistent with \eqref{WAIRY}, i.e. for $H(x, \xi) \leq E$ the argument of the Airy function is negative and therefore the asymptotics
are oscillatory.

  Next, we remark that the result of Theorem \ref{T:SMOOTH-Prop} is stated and a proof is sketched  in  \cite{O98,TL}. 
  Our motivation for presenting a rigorous proof is to generalize our earlier result \eqref{WAIRY} and to relate
  these special asymptotics to other types of asymptotics expansions for Wigner functions. The argument of
  the Airy function and the
leading coefficient and argument of the Airy function agree with the calculations in \cite{O98,TL}.  To connect the notations,  in those articles,   the leading coefficient is  given by $A_0(x, \xi, E) S_0^{1/6}$ where $A_0$ in  \cite[(17)]{TL} 
is the coefficient $|\frac{dt}{dE} \det (1 + M)|^{-\half}$.
Also, $S_0^{2/3} $ is the argument of Airy in \cite{TL}. In our notation, which follows that of \cite{HoI}, the  Airy argument is 
$\rho$ and  $\rho^{1/4}=  (\rho^{3/2})^{\frac{1}{6}} =  (S_0)^{1/6}.$

Further, our proof of Theorem \ref{T:SMOOTH-Prop} relies on a stationary phase with fold singularities (see \cite[Page 236]{HoI}) from which we see that when  $\rho (x, \xi) \approx 0$ Theorem  \ref{T:SMOOTH-Prop}  gives an asymptotic expansion in powers of $\hbar^{1/3}$ that remains valid in any region of $T^* \R^d$ in which $\hbar^{-\frac{2}{3}} \rho (x, \xi)$ stays bounded. However,  due to the exponential decay of the Airy functions on the positive axis, $\rho(x, \xi)$ becomes very small when  $\hbar^{-\frac{2}{3}} \rho (x, \xi)$  is positive, and stationary
 phase asymptotics given below in Theorem \ref{T:SMOOTH-Prop2} then become valid.

The Airy scaling asymptotics of  $W_{\hbar, f,1 ,E}(x, \xi) $ for $(x, \xi)$ in the boundary layer around $\Sigma_E$ (i.e. $\hbar^{1/3} $ close
to $\Sigma_E$)  are due to a fold singularity in the map \eqref{PSItDEF}. More precisely, 
in   Section \ref{EFOLDSECT}, it is explained that  the  Airy asymptotics are due to the fact that the  relevant Lagrangian submanifold $\Lambda_{E} \subset T^*(T^* \R^d)$ in the second cotangent bundle (defined by \eqref{LAGDEF2})
has a fold singularity  around $\Sigma_E$,  i.e.  $\Sigma_E$ is a caustic for the Wigner function (see Section \ref{EFOLDSECT}).  The proof of Theorem \ref{T:SMOOTH-Prop} is based on the use of the Herman-Kluk propagator and on   the Airy asymptotics  results of Chester-Friedman-Ursell \cite{CFU} as given 
in \cite[Section 7]{HoI} and \cite{GSt}.
 Theorem \ref{T:SMOOTH-Prop} is a rigorous version of the result stated in \cite[(7.21)]{O98}. It does not seem to appear in \cite{Ber89}.

Finally, we conjecture that when the smooth test function $f$ is replaced by the indicator function of the spectral interval, then  the sharp Wigner-Weyl sum has the asymptotics,
$$\sum_{j: |E_j(\hbar) - E | \leq C \hbar} W_{\psi_j, \psi_j}(x, \xi) = \hbar^{-d} 
\left( \hbar^{\frac{1}{3}} \rm{Ai}(- \hbar^{-\frac{2}{3}} \rho(x, \xi) ) \sum_{\nu=0}^{\infty} u_{0, \nu}(x, \xi) \hbar^{\nu} \right)
+ O(\hbar^{-d + \frac{2}{3}}). $$
Classically such asymptotics for spectral intervals are obtained from the smoothed results \eqref{AIRY} by applying  cosine Tauberian theorems.  But   Tauberian theorems make the hypothesis  that the terms of the sums are non-negative, whereas it is a well-known and important phenomenon that Wigner functions are almost never globally positive. This raises the question whether  
$W_{ \psi_j, \psi_j}(x, \xi) \geq 0$ for $(x, \xi) $ in an $\hbar^{2/3}$ tube around $\Sigma_E$ when the eigenvalue of $\psi_j$ satisfies $|E_j(\hbar) - E | \leq C \hbar$. This is plausible, since asymptotically the Airy function is only evaluated where it is positive; 
moreover, in a weak sense, the Wigner functions tend to a delta function on $\Sigma_E$. But this is far from sufficient to
prove positivity.  A possible source of counter-examples where $W_{\psi_j, \psi_j}(x, \xi)$ can be  negative at some points of $\Sigma_E$ is the Wigner distribution of product eigenstates  of the isotropic
harmonic oscillator, \begin{equation}
\phi_{\alpha,h}(x)=h^{-d/4}p_{\alpha}\lr{x\cdot
h^{-1/2}}e^{-x^2/2h},\label{E:Scaling Relation}
\end{equation}
where $\alpha=\lr{\alpha_1,\ldots, \alpha_d}\geq (0,\ldots,0)$ is a
$d-$dimensional multi-index and
$p_{\alpha}(x)$ is the product $\prod_{j=1}^d p_{\alpha_j}(x_j)$ of the
hermite polynomials $p_k$ (of degree $k$)
in one variable. There are similar products for generic oscillators, where the multiplicity of each eigenvalue equals one.   The Wigner distribution is the product of those of the factors and are given by products in the variables $(x_1, \dots, x_d)$ of one-dimensional  Laguerre polynomials of the variables $x_j^2$.   The product is positive if and only if the number of negative factors is even.  It seems non-obvious whether one can construct products with an odd number of negative factors when $ |E_j(\hbar) - E | \leq C \hbar$ and $\sum_{j=1}^d x_j^2 = E$. Another source of counter-examples could come from Wigner distributions of coherent states along elliptic periodic orbits. In the case of the isotropic harmonic oscillator, the Wigner
distribution is positive near the periodic orbit on $\Sigma_E$ but appears to be negative at some points on $\Sigma_E$ away from the orbit \cite{Lo}. As far as we know, the study of negativity sets of Wigner distributions of eigenfunctions under the above constraints has never been studied.

There are many further scaling asymptotics results for the smoothed and sharp Weyl-Wigner sums \eqref{WEYLWIGNER}-\eqref{Whfdelta}.
In Section \ref{OTHER} we review some model results for the isotropic harmonic oscillator for various types of spectral intervals that should admit generalizations to \eqref{Hh}. We note
that the asymptotics results for the longer spectral intervals involve integrals of the Airy function rather values of the Airy function
as for intervals of order $O(\hbar)$.

\subsection{Pointwise Semi-Classical Asymptotics of Wigner Transform of the Propagator}
As a interesting and useful warm-up to the proof of Theorem \ref{T:SMOOTH-Prop}, we first give a new proof of what is essentially a well-known result from the physics literature (e.g. \cite{Ber89}) giving pointwise semi-classical asymptotics for the Wigner transformation of the \Sch operator $\widehat{H}_\hbar$. Before stating the precise result in Proposition \ref{T:Wig-Prop}, we recall some notation. First, the semi-classical Wigner transform is defined to be  the unitary operator 
$$\wcal_{\hbar}: L^2(\R^d \times \R^d) \to L^2(T^*\R^d),$$ 
taking Hilbert-Schmidt kernels  $K_{\hbar} \in L^2(\R^d \times \R^d)$ to their Wigner distributions. The  semi-classical Wigner transform extends to temperate  (i.e. Schwartz)   distributions $K_{\hbar} \in \scal'(\R^d \times \R^d)$ in the dual space of Schwartz space by,
\begin{equation}
 \wcal_{\hbar}(K_{\hbar})(x, \xi): = \int_{\R^d} K_{\hbar} \left( x+\frac{v}{2}, x-\frac{v}{2} \right) e^{-\frac{i}{\hbar} v \xi} \frac{dv}{(2\pi h)^d}. \label{E:WignerK}
\end{equation} 
In particular, each Wigner function $W_{\psi_j,\psi_j}$ appearing in Theorem \ref{T:SMOOTH-Prop} can be written as the Wigner transform
\[
W_{\psi_j,\psi_j}(x,\xi)=  \wcal_{\hbar}(\Pi_{j,\hbar})(x, \xi),
\]
of the kernel
\[
 \Pi_{j,\hbar}(x,y) :=\psi_j(\hbar, x)\psi_j(\hbar, y) 
\]
of the rank one projection onto the state $\psi_j(\hbar,\cdot)$. One advantage of viewing Wigner functions in this way is that, while Wigner functions $W_{\psi_j,\psi_j}$ are quadratic in $\psi_j$, the Wigner transform $\wcal_{\hbar}(K_\hbar)$ is linear in $K_\hbar$. We will make use of this when proving Theorem \ref{T:SMOOTH-Prop}. But first we consider the  `propagator'   $U_{\hbar}(t) = e^{-\ihbar t \widehat{H}_h}$, i.e. the solution operator of the Cauchy problem for the  \Sch equation
$$i \hbar \frac{\partial}{\partial t} u = \widehat{H}_{\hbar} u.$$
More precisely, we will consider the Schwartz kernel
\begin{equation} \label{PROPDEF} 
U_h(t, x,y) =e^{-\ihbar t \widehat{H}_h}(x,y) = \sum_{j\geq 0} e^{-\ihbar tE_j(\hbar)} \Pi_{j,\hbar}(x,y)
\end{equation}
of $U_{\hbar}(t)$ and its Wigner transform 
\begin{equation} \label{WIGU} 
\ucal_\hbar(t,x,\xi):= \wcal(U_\hbar(t,\cdot,\cdot))(x,\xi)=\int e^{- \ihbar 2 \pi  \xi\cdot v} U_{\hbar} \left(t, x + \frac{v}{2}, x - \frac{v}{2}\right) \frac{dv}{(2 \pi \hbar)^d}.
\end{equation}

Following Berry \cite{Ber89}, our first task will be to obtain, for fixed $(t,x,\xi)$, semi-classical pointwise asymptotics for $\ucal_\hbar(t,x,\xi)$. Specifically, we prove the following formula for pointwise semi-classical asymptotics of the Wigner function of the propagator,  stated  in a heuristic way by M.V. Berry \cite[(21)]{Ber89}   (see also  \cite{TL}  and \cite{O98}.
\begin{prop}\label{T:Wig-Prop}
Let $\widehat{H}_\hbar$ satisfy  \eqref{SUBQ}. Let $(t, x, \xi) $ be fixed and consider the set  of points $(q_j, p_j) \in T^* \R^d$ and
the Hamiltonian arcs $\gamma_j = \gamma_{j; t, x, \xi}$ such that $\gamma_{j; t, x, \xi} (0)= (q_j, p_j)$ and 
\begin{equation} \label{xxi} x=\frac{1}{2}\lr{\gamma_{j;t,x,\xi}(0)+ \gamma_{j;t,x,\xi}(t)},\qquad \xi = \frac{1}{2}\lr{\dot{\gamma}_{j;t,x,\xi}(0) + \dot{\gamma}_{j;t,x,\xi}(t)}.
\end{equation}
Consider the Jacobian
\[
M_j(t,x,\xi):=D_{q,p}\Phi^{t}(q,p)\big|_{(q,p)=(q_j(t,x,\xi),p_j(t,x,\xi))}
\]
of the Hamilton flow $\Phi^t$ of the endpoint  $\gamma_{j;t,x,\xi}(t)=\Phi^{t}(q_j,p_j)$ with respect to initial point $(q,p)$, and 
assume that  
\[
\det(I + M_j(x,\xi)) \not= 0,\qquad \forall j.
\]
Then, 
\eqref{WIGU} admits the pointwise semi-classical asymptotics,
\[\ucal_\hbar(t,x,\xi)=2^d  \sum_j \frac{\exp\lr{i \left[\frac{S_j}{\hbar}+\eta_j\right]}}{\det^{1/2}(I + M_j(x,\xi))} + O(\hbar) ,\]
where  $S_j$ is the classical action and $\eta_j$ is the Maslov index associated to the path $\gamma_{j; t, x, \xi}$ (Section \ref{BACKGROUND}).
\end{prop}
This Proposition is essentially well known in the physics literature (e.g. \cite{Ber89} and \cite{O98}). We prove it in a new way  using the Herman-Kluk parametrix for the propagator (Section \ref{HKSECT}). The basic idea, explained in \S \ref{HKSECT-WIGPROP-PROOF}, of the proof of Proposition \ref{T:Wig-Prop} is to use the Herman-Kluk parameterix \cite{R10} for the kernel of the propagator $U_\hbar(t,x,y)$ to obtain a parametrix for its Wigner distribution of the form
\begin{equation}\label{E:ucal-para-intro}
\ucal_\hbar(t,x,\xi)\sim \int_{\R^{2d}}e^{\ihbar \Psi(q,p,v;t,x,\xi )}a_\hbar(t,q,p)dqdp,    
\end{equation}
where $\Psi$ is an explicit complex-valued phase function depending on the underlying the classical Hamiltonian flow (see \eqref{PSI}) and $a_\hbar$ is a polyhomogeneous symbol (see \eqref{AMP}). This parametrix is valid for long times due to the sub-quadratic assumption \eqref{HYP} on our Hamiltonians. A straight forward stationary phase argument, detailed in \S \ref{HKSECT-WIGPROP-PROOF}, yields Proposition \ref{T:Wig-Prop}.

\subsection{Outline of the Proof of Theorem \ref{T:SMOOTH-Prop}} We will deduce Theorem \ref{T:SMOOTH-Prop} by analyzing the following relation between $W_{\hbar, f,1,E}$ and the Wigner function of the propagator:
\begin{equation}\label{E:WU}
W_{\hbar, f,1,E}(x,\xi) = \int_{\R}\widehat{f}(t)e^{itE/\hbar}\ucal_\hbar(t,x,\xi)\frac{dt}{2\pi},    
\end{equation}
which is a consequence of Fourier inversion and the linearity of the Wigner transform. As discussed in detail in \cite{HZ20,HZb}, $U_\hbar(t,x, y)$ is not locally $L^1$ but is nonetheless well-defined as a tempered distribution in the sense that integral of the form
\begin{equation}\label{ucalhf}  
\ucal_{\hbar, f}(x, \xi) : =\int_{\R} \hat{f}(t)  \ucal_{\hbar}(t, x, \xi)dt 
\end{equation}
are well-defined for $f \in \scal(\R).$  It therefore has a well-defined distributional Wigner transform, making \eqref{E:WU} a valid expression. To analyze \eqref{E:WU}, just as in the derivation of Proposition \ref{T:Wig-Prop}, we start with the Herman-Kluk \cite{R10} parameterix to obtain a parametrix for the Wigner transform of the propagator as in  \eqref{E:ucal-para-intro}. In combination with \eqref{E:WU} this yields an oscillatory integral representation 
\begin{equation}\label{E:Wigner-FIO}
W_{\hbar,f,1,E}(x,\xi)\sim \int_{\R^{3d}}e^{\ihbar (\Psi(q,p,v;t,x,\xi)+tE)}a_\hbar(t,q,p)\frac{dqdpdv}{(2\pi \hbar)^d}\frac{dt}{2\pi}.    
\end{equation}
To obtain the asymptotic expansion in Proposition \ref{T:Wig-Prop} we apply stationary phase. The key point is that the critical points with respect to $p,q,v$ are non-degenerate (as they were in the proof of Proposition \ref{T:Wig-Prop}) but that the critical points with respect to $t$ are degenerate with a fold singularity at the only critical point in the support to $f$, namely at $t=0$.

After integrating out the $p,q,v$ variables, to obtain asymptotics for $W_{\hbar, f, 1, E}(x, \xi)$ we mst integrate in $t$. Theorem \ref{T:SMOOTH-Prop} follows by applying the Chester-Friedman-Ursell (or Malgrange preparation theorem) to the phase of the $dt$ integral resulting from Proposition \ref{T:Wig-Prop}. In Section \ref{MALGRANGE}, we use the fold singularity at $t=0$ to complete the proof of Theorem \ref{T:SMOOTH-Prop}. Leveraging Theorem \ref{T:SMOOTH-Prop2} and Proposition  \ref{T:Wig-Prop}, we will see in Proposition \ref{PROPPROPt} that $W_{\hbar, f,1 ,E}(x, \xi) $ is a  semi-classical Fourier integral kernel with a complex phase. The critical points in $ (t, q,p, v)$ at which the imaginary part of the phase vanishes (and hence the phase integrand is not exponentially small in $\hbar$) are all solutions of
\[
\Theta^t(q,p)=(x,\xi),\qquad  H(q, p) = E.
\]
The singularities of \eqref{E:Wigner-FIO} are determined geometrically by the following result, which is proved at the same time as Proposition \ref{FOLDPROP} below.
\begin{prop}\label{FOLDPROPintro}  The map $\Theta^*$ defined in \eqref{PSItDEF} has a fold singularity along $\{0\} \times \Sigma_E.$ That is, \eqref{PSItDEF} fixes the diagonal when $t=0$ and the kernel
of its derivative is spanned by the vector field $\frac{\partial}{\partial t}$.
\end{prop}

Proposition \ref{FOLDPROPintro} shows that the Lagrangian  generated by the phase in \eqref{E:Wigner-FIO} has a fold singularity for critical points at $t=0$. To analyze it, we rely on variants of the Malgrange preparation theorem for fold singularities, in the stronger form given in \cite{CFU, L61}, to put the phase into cubic normal form.  See  Proposition \ref{HORMPROP} and Proposition \ref{HORMPROP2}. We note that the latter  statements assume that the phase is real valued, while in  \cite{CFU,L61} the phase and amplitude are assumed to be complex analytic. Neither assumption holds for the phase  of the spectral  Wigner distribution in the generality of this article. For this reason, we first integrate out all but the time $t$ variable in Section \ref{HKSECT} (see in particular Section \ref{ONEDRED}), so that we may use Proposition \ref{HORMPROP}.


\subsection{Contributions to Smooth Wigner Functions Coming from Non-degenerate Critical Times} Although we do not use it in the proof of Theorem \ref{T:SMOOTH-Prop},
we state a second known result on the asymptotics of \eqref{Whfdelta} for which the stationary phase method applies. We only use it to make comparisons with the Airy asymptotics as $(x, \xi)$ moves away from the fold singularity. In the next
Proposition we retain the notation of Proposition \ref{T:Wig-Prop}.

\begin{prop}\label{T:SMOOTH-Prop2}
Let $H_\hbar$ satisfy \eqref{SUBQ}. Fix $(E, x, \xi)$ with $(x, \xi) \in B_E$ and consider all solutions  $(t_j, q_j, p_j)$ of  \eqref{xxi}
for which $(q_j, p_j) \in \Sigma_E$ and $t_j \in \rm{supp} \;\hat{f}$.  Assume that  $0\not \in \mathrm{supp}(\hat{f})$ and that $\frac{dt_j}{dE} \det (1 + M_j(x, t_j(E)) \not= 0$  for all $j$.  Then the smoothed spectral Wigner function  \eqref{Whfdelta}  admits the pointwise semi-classical asymptotics, 
$$W_{\hbar, f, 1,E}(x, \xi): =
 \frac{2^{d+1}}{\sqrt{2 \pi \hbar}} \sum_j \hat{f} (t_j) \left| \frac{dt_j}{dE} \det (1 + M_j(q_j, p_j,  t_j(E))^{-1} \right|^{\half} \cos \left( \frac{S_j(x, \xi, E)}{\hbar} + m_j\right) + O(\hbar^{1/2}). $$ 
 Here, $S_j(x, \xi, E)$ is the action along the trajectory $\gamma_{j; t, x, \xi}$.  
\end{prop}
Under the assumption \eqref{HYP}, there are no solutions $(t_j, q_j, p_j)$  unless $(x, \xi) \in B_E$. Moreover, 
as explained in \cite{Ber89}, 
if the trajectory through $(q, p) $ is  periodic of period $T$, and if $(t, q, p)$ is a solution of \eqref{xxi}, then
$(t + T, q, p)$ is another solution, but the non-degeneracy condition is not satisfied.

\subsection{\label{OTHER}Further  scaling results} The results of M.V. Berry \cite{Ber89} pertain mainly to the the contribution of periodic orbits to the Wigner spectral
asymptotics. It is evident from \eqref{SOLVE} that if $(q, p)$ is a periodic point of period $T$ then one gets further
solutions by replacing $(q, p)$ by $\Phi^T(q, p)$. The fold singularity along $\Sigma_E$ at $t=0$ also occurs
at time $t= T$ at the periodic points of $\Sigma_E$ of period $T$. The order of the asymptotics is higher than it would
be for non-degenerate critical points, and Berry therefore referred to the caustic enhancement of such periodic orbits
as `scars'.  The methods of this article extend to the periodic orbit case with few modifications but for the sake of
brevity we do not include them here. 

In \cite{HZb},  further scaling asymptotics are proved in the case of the isotropic harmonic oscillator. The essential difference to
Theorem \ref{T:SMOOTH-Prop} is that much larger spectral intervals are assumed. We do not generalize these results to 
general Hamiltonians satisfying \eqref{SUBQ} or \eqref{HYP} in this article,  but it is very likely that generalizations do exist. We briefly review
these additional scaling results. 

Instead of spectral intervals of width $\hbar$ we consider intervals of width $\hbar^{2/3}$, e.g. 
$[E - \lambda_- \hbar^{2/3}, E + \lambda_+ \hbar^{2/3} $. In the smoothing asymptotics we let $\delta = \frac{2}{3}$ and consider
$$W_{\hbar, f, \frac{2}{3},E}(x, \xi): =  \sum_{j=1}^{\infty}  f\lr{\hbar^{-\frac{2}{3}}(E- E_j(\hbar))}W_{ \psi_j, \psi_j}(x, \xi).$$  We define the  rescaled variable $u=u(x,\xi)$ centered at the energy surface $\Sigma_E$ by  
\[H(x, \xi)  =E+ u  \lr{\hbar/2E}^{2/3}.\] We then prove that,
$$ W_{2/3,E,\lambda_{\pm}}(x,\xi)=(2\pi\hbar)^{-d}C_E\int_{-\lambda_+}^{-\lambda_-}\Ai\lr{\frac{u}{E}+\lambda C_E} d\lambda + O(\hbar^{-d+1/3-\delta}),\qquad C_E = (E/4)^{1/3}. $$
Instead of getting the value of the Airy function at the scaled parameter $u$, we get an integral over the values due to the larger
spectral interval.

   Furthermore, in   the even larger interval $I_{\hbar} = [0, E]$,
 for any $\epsilon>0$, we obtained the bulk asymptotics, 
\begin{equation} \label{BULKSCALINGCOR} 
W_{\hbar, [0, E]}(x, \xi)  =   \lr{2\pi \hbar}^{-d}  \left[\int_0^{\infty} \Ai\lr{\frac{u}{E}+\lambda} d\lambda +O(\hbar^{1/3-\epsilon}\abs{u}^{1/2})+ O(\abs{u}^{5/2}\hbar^{2/3-\epsilon})\right],
\end{equation}
where the $O$-symbol depends only on $d,\epsilon.$


The articles  \cite{Ber89, O98} use the special (Lorentzian) test function $f_{\epsilon}(E - E_j(\hbar)) = \delta_{\epsilon}(E - E_j(\hbar)): =  -\frac{1}{\pi} \Im (E + i \epsilon - x)^{-1} $, with various choices of $\epsilon$ stated in \cite[Pages 220-221]{Ber89} as ranging from the mean level spacing of order $\hbar^d$ and the semi-classical scaling $\frac{\hbar}{T_{\min}}$ where $T_{\min}$ is the length of the shortest periodic orbit on $\Sigma_E$. The same test functions and energy scales are used in \cite[Section 7]{O98}. The mathematical techniques of this article (along with other spectral asymptotics articles in the mathematics literature) are not valid below the semi-classical scale, and in particular do not give scaling results on the length scale $\delta  = \hbar^{\nu}$ with $\nu > 1$; they do apply on the scale  $\frac{\hbar}{T_{\min}}$, and we prove Theorem \ref{T:SMOOTH-Prop} by using a special case of the  Malgrange preparation theorem due to  Chester-Friedman-Ursell \cite{CFU} to put the phase into normal form; the relevant theorems are   explained in detail in  \cite[Chapter 7]{HoI} and in \cite[Page 439, Page 444]{GSt}. Such asymptotics are used to determine the asymptotics of oscillatory integrals whose phase functions exhibt  fold singularities, such as occur in diffraction theory. In Section \ref{EFOLDSECT} we identify the relevant folding map.

Besides linking the somewhat heuristic asymptotics calculations of \cite{Ber89, O98} to the mathematical literature, in particular making
more precise the scale of the asymptotics,  a novelty of
our presentation 
is to use the Herman-Kluk propagator as discussed by D. Robert \cite{R10} to construct a `parametrix' for \eqref{Whfdelta}. 

\subsection{Acknowledgements} Thanks to Mike Geis and Nick Lohr for many helpful comments that have improved
the exposition.

\section{\label{BACKGROUND} Background on Classical mechanics} In this section, we recall some basic results from classical mechanics and set notation. With $\Phi^t$ the Hamiltonian flow defined in \eqref{PHIt}, we continue to denote its derivative at $(q, p)$ by  
\begin{equation} \label{MtDEF} D_{q, p} \Phi^t: = M_t=\twomat{A_t}{B_t}{C_t}{D_t},\end{equation}
where
\[A_t = \dell_q q_t,\quad B_t = \dell_p q_t, \quad C_t =\dell_q p_t,\quad D_t = \dell_pp_t,\]
with rows corresponding to components and columns to derivatives:
\[\lr{\dell_{q_j}q_t}_k = \lr{A_{t}}_{k,j}.\]
Since $\Phi^t$ is a Hamiltonian flow, we have
\[\Omega = M_t^T \Omega M_t\]
where $\Omega= \twomat{0}{I}{-I}{0}$ is the symplectic form \cite{F} and we therefore have
\begin{align}
  A_tC_t = C_t^TA_t ,\qquad A_t^TD_t - C_t^TB_t = \Id,\qquad B_t^TD_t = D_tB_t\label{E:symplectic-jacobian}.
\end{align}
The action along the $\Phi^t$ orbit with initial value $(q,p)$ is defined by,
\begin{equation} \label{Stqp} S(t,q,p):=\int_0^t (\dot{q}_s\cdot p_s - H(s,q_s, p_s)) ds\end{equation}
We only deal with autonomous Hamiltonians, for which $H(s, q_s, p_s) = H(q_s, p_s)$ and then $H(q_s, p_s)$
is constant along Hamilton orbits and the second term is $t H(q,p)$. We will have occasion to use the following elementary result.

\begin{lemma} \label{DERIVS} We have
  $$\left\{ \begin{array}{l} \dell_q S(t,q,p)= \lr{\dell_q q_t}\cdot p_t - p, \\ \\ 
  \dell_p S(t,q,p)=\lr{\dell_p q_t}\cdot p_t, \\ \\ 
\dell_t S(t,q,p)=\dot{q}_t \cdot p_t - H(q_t, p_t) . \end{array} \right. $$
\end{lemma}

\begin{proof}
  Using the equations of motion $\dell_p H = \dot{q}_s$ and $\dell_q H = - \dot{p}_s$, we get
   \begin{align*}
\dell_qS(t,q,p)&=\dell_q\lr{\int_0^t \dot{q}_s\cdot p_s - H(q_s, p_s)ds}\\&= \int_0^t\lr{\dell_q \dot{q}_s\cdot p_s + \dot{q}_s \cdot \dell_q p_s - \dell_q H(q_s,p_s)\dell_q q_s - \dell_p H(q_s,p_s) \dell_q p_s}ds \\
&= \int_0^t \langle \dell_q \dot{q}_s,  p_s  \rangle+ \langle  \dot{q}_s, \dell_q p_s  \rangle+ \langle \dot{p}_s, \dell_q q_s
\rangle - \langle \dot{q}(s),  \dell_q p_s \rangle ds\\ & 
= \int_0^t \langle \dell_q \dot{q}_s,  p_s  \rangle + \langle \dot{p}_s, \dell_q q_s
\rangle  ds 
 = \int_0^t \partial_s \langle \dell_q q_s,   p_s \rangle ds = \langle \dell_q q_s,   p_s \rangle |_0^t,
 \end{align*}
 proving the first statement. Similarly,

  \begin{align*}
\dell_pS(t,q,p)&=\dell_p\lr{\int_0^t \dot{q}_s\cdot p_s - H(q_s, p_s)ds}\\
&=\int_0^t\lr{\dell_p \dot{q}_s\cdot p_s + \dot{q}_s \cdot \dell_p p_s - \dell_q H(q_s,p_s)\dell_pq_s - \dell_p H(q_s,p_s) \dell_p p_s}ds \\ & = 
\int_0^t \left( \dell_p \dot{q}_s\cdot p_s + \dot{q}_s \cdot \dell_p p_s  + \langle \dot{p}_s, \dell_pq_s \rangle - \langle \dot{q}_s,  \dell_p p_s \rangle \right) ds \\ & = 
\int_0^t \left( \dell_p \dot{q}_s\cdot p_s   + \langle \dot{p}_s, \dell_pq_s \rangle \right) ds = \int_0^t \frac{d}{ds}
\langle p_s, \dell_p q_s \rangle ds = \langle p_s, \dell_p q_s \rangle |_0^t, 
  \end{align*}
  proving the second statement. The third statement is obvious.

\end{proof}

\section{\label{HKSECT-WIGPROP-PROOF} Herman-Kluk Parametrix and the Proof of  Proposition \ref{T:Wig-Prop}}
In this section, we prove Proposition \ref{T:Wig-Prop}. For this, we recall in \S \ref{HKSECT} the Herman-Kluk Parametrix for the propagator. Then we use this parameterix in conjunction with stationary phase to complete the proof of Proposition \ref{T:Wig-Prop} in \S \ref{S:Wig-Prop-Proof}. We then provide in \S \ref{dv} and \ref{S:Rem-Lag} two conceptual remarks on the proof that will be useful for guiding our subsequent proof of Theorem \ref{T:SMOOTH-Prop}. 

\subsection{\label{HKSECT} Herman-Kluk Parametrix for the Wigner Transform of the Propagator} As explained in \cite{R10, RS} (see \cite[(1.7)]{R10}), for subquadratic Hamiltonians satisfying \ref{SUBQ}, one can construct  to construct a long-time parametrix 
\begin{equation} \label{HKPROP} U_\hbar(t,x,y) \sim \int_{\R^{2d}}e^{\frac{i}{\hbar}\Psi_{HK}(t,q,p;x,y)}a_\hbar(t,q,p)dqdp\end{equation}
for the propagator $U_\hbar(t,x,y)$ (see \eqref{PROPDEF}) with the complex Herman-Kluk phase,
\begin{equation} \label{PSIHK} \Psi_{HK}(t, q,p; x,y):=S(t,q,p)+p_t\lr{x-q_t} - p(y-q)+\frac{i}{2}\lr{\abs{x-q_t}^2 + \abs{y-q}^2}.
\end{equation}
Here, $S(t,q,p)$ is the action given by \eqref{Stqp}. The phase $\Psi_{HK}$ is called a positive complex phase since its imaginary part is positive. Moreover, $a\sim \sum_j \hbar^j a_j(q,p)$ is a polyhomogeneous symbol with
\begin{equation}\label{AMP} a_0=\lr{\det\lr{A_t + D_t + i(B_t-C_t)}}^{1/2}\exp(-itH(q,p))\end{equation}
where we've written as in \eqref{MtDEF} Section \ref{BACKGROUND}
\[A_t=\frac{\dell q_t}{\dell q},\quad B_t=\frac{\dell q_t}{\dell p},\quad C_t=\frac{\dell p_t}{\dell q},\quad D_t=\frac{\dell p_t}{\dell p}.\]
The parametrix is valid for all times $t<\log (1/\hbar)$ and the symbol estimates are uniform over compact intervals in $t.$ The parameterix \eqref{HKPROP} for the propagator $U_\hbar(t,x,y)$ reveals that its Wigner transform $\ucal_\hbar(t,x,\xi)$ is an oscillatory integral 
with positive complex phase. Namely, 
\begin{equation}\label{E:Ucal-PSI}
\ucal_\hbar(t,x,\xi) \sim (2\pi \hbar)^{-d}\int_{\R^{3d}} e^{\frac{i}{\hbar}\Psi(t,x,\xi;q,p,v)} a_\hbar(t,q,p)\frac{dqdpdv}{(2\pi \hbar)^d},    
\end{equation}
where
\begin{align}
    \label{PSI} \Psi(t, x, \xi; q, p, v) &:=S_t+p_t\lr{x+\frac{v}{2}-q_t}-p\lr{x-\frac{v}{2}-q}-v\xi\\
    \notag &\qquad +\frac{i}{2}\lr{\abs{x+\frac{v}{2}-q_t}^2+\abs{x-\frac{v}{2}-q}^2}. 
\end{align} 
The integral is of course exponentially small in $\hbar$ unless the imaginary part vanishes,  i.e. 
\begin{equation} \label{IMPART0} 
\Im \Psi = \abs{x+\frac{v}{2}-q_t}^2+\abs{x-\frac{v}{2}-q}^2=0. 
\end{equation}
\subsection{Proof of Proposition \ref{T:Wig-Prop}}\label{S:Wig-Prop-Proof}

We are now ready to complete the proof of Proposition \ref{T:Wig-Prop}. In the notation of Section \ref{BACKGROUND}, the main technical result is the following. 

\begin{proposition}\label{PROPPROP}
  The Wigner function $\ucal_\hbar(t,x,\xi)$ of the propagator is a semi-classical Fourier integral kernel with a positive complex phase. The dominant critical points (in $ (q,p, v)$) at which \eqref{IMPART0} holds are solutions of
\[
x = \frac{q+q_t}{2},\quad \xi = \frac{p+p_t}{2},\quad v= q_t-q.
\]
Moreover, there exist $\epsilon,\epsilon_0>0$ such that for $(x, \xi)$ in an $\epsilon$-neighborhood of $\Sigma_E$ and $t \in (-\epsilon_0, \epsilon_0)$, critical points are uniformly non-degenerate. The critical value of the phase at a critical point $(q, p,v) $ is given by:
\begin{equation} \label{PSIc} 
\Psi_c(t, x, \xi) := \int_{\gamma_{t, x, \xi}} (pdq - H ds)  -\langle (q_t - q),  \xi \rangle, 
\end{equation}
where $\gamma_{t, x, \xi}$ was defined in \eqref{gammaDEF} as  the unique Hamiltonian arc (up to the involution \eqref{SYM}) with $|t|$ small on $\Sigma_E$  such that $ (x, \xi) = \half (\gamma_{t, x, \xi}(0) + \gamma_{t, x, \xi}(t))$. 
\end{proposition}
Define the action,
\begin{equation} \label{ACTIONDEF} S(t, x, \xi) = \int_{\gamma_{t, x, \xi} } (pdq - H ds) \end{equation}
along the path $\gamma_{t, x, \xi}$.

\begin{proof}[Proof of Proposition \ref{PROPPROP}]
As mentioned above, the imaginary part of the phase is strictly positive unless 
\begin{equation} \label{CPEa} x+\frac{v}{2}-q_t=x-\frac{v}{2}-q=0.\end{equation}
Hence, any critical points where these conditions are satisfied will provide the dominant contribution to the Wigner function. 
It is immediate from \eqref{CPEa} that at the dominant critical points, 
\begin{equation} \label{CPE} x=\frac{q+q_t}{2},\;\; \qquad v = q_t-q.\end{equation}
In addition, the equation $d_v\Psi=0$ implies that 
\[  \xi = \frac{p+p_t}{2}.\]
\end{proof}
To complete the proof of Proposition \ref{T:Wig-Prop}, we fix $t > 0$,  consider the full $dq dp dv$ integral \eqref{E:Ucal-PSI}, and apply the method of stationary phase for positive complex phase functions (cf. \cite[Theorem 7.7.5]{HoI}). For this, we need to compute the Hessian of \eqref{PSI} at the dominant critical points. We use the notation \eqref{MtDEF}, which we repeat for the reader's convenience,
\[
A_t(q,p) = \dell_q q_t(q,p),\quad B_t(q,p) = \dell_pq_t(q,p),\quad C_t(q,p) = \dell_qp_t(q,p),\quad D_t(q,p)=\dell_p p_t(q,p).
\]
The key result is the following computation. 
\begin{lem}\label{P:det-Hess} 
Fix $t,x,\xi$ and suppose $(q_c, p_c, v_c)$ is a dominant critical point of $\Psi(q,p,v;t,x,\xi)$ \eqref{PSI}  in the variables $(q, p, v)$ for which $\dell_p q_t \neq 0.$ Then
 $\det\lr{\text{Hess}(\Psi)}(q_c,p_c,v_c)$ is given by,
\[(-1)^d\cdot \det(1+M_t(q_c, p_c))\det\lr{A_t(q_c,p_c) + D_t(q_c,p_c) +i (B_t(q_c, p_c)-C_t(q_c, p_c))}.\]
\end{lem}
\begin{proof}

We claim that at the critical point,
 \begin{align*}
    \dell_{pp} \Psi &= B_t^T\lr{-D_t + i B_t},\qquad & \dell_{pq} \Psi = B_t^T\lr{-C_t + i A_t}\\
\dell_{pv} \Psi &= \frac{1}{2}\lr{\Id+ D_t - iB_t}^T,\qquad    & \dell_{vv} \Psi = \frac{i}{2}\Id\\
\dell_{qq} \Psi &= -A_t^TC_t + i\lr{A_t^TA_t + \Id},\qquad  & \dell_{qv} \Psi = \frac{1}{2}\lr{C_t+ i \lr{-A_t + \Id}}.
  \end{align*}
We summarize the key points of the calculations. Apriori, the Hessian involves second derivatives in $(q, p)$ of $(q_t, p_t, \dot{q}_t, \dot{p}_t)$ but in fact these cancel at a dominant critical point. First, 
  \begin{equation}
\dell_{p_j}\Psi = \dell_{p_j}p_t\lr{x+\frac{v}{2} - q_t} - \lr{x_j-\frac{v_j}{2}-q_j} + i \lr{-\dell_{p_j}q_t\lr{x+\frac{v}{2}-q_t}},\label{E:p-deriv}
\end{equation}
and at a critical point, \[\dell_{p_k,p_j}\Psi = -\dell_{p_j}p_t \dell_{p_k}q_t + i\lr{\dell_{p_j}q_t \del_{p_k}q_t}=(-\dell_{p_j} (p_t+i\dell_{p_j}q_t)\dell_{p_k}q_t=\lr{(-D_t+iB_t)^TB_t}_{k,j}\]
Also,
  \begin{equation}
\dell_{q_j}\Psi = \dell_{q_j}p_t\lr{x+\frac{v}{2} - q_t} - i \lr{\dell_{q_j}q_t\lr{x+\frac{v}{2}-q_t}+\lr{x_j -\frac{v_j}{2}-q_j}},\label{E:q-deriv}
\end{equation}
and at a critical point, 
\[
\dell_{p_k,q_j}\Psi = -\dell_{q_j}p_t \dell_{p_k}q_t + i\lr{\dell_{q_j}q_t \del_{p_k}q_t}.\]
Starting from \eqref{E:p-deriv}, we have
\[\dell_{v_kp_j}\Psi = \frac{1}{2}\lr{(\dell_{p_j}p_t)_k + \delta_{j,k} - i (\dell_{p_j}q_t)_k}=\frac{1}{2}\lr{\Id+D_t-iB_t}_{k,j}.\]
Again starting from \eqref{E:p-deriv}, we have
\[\dell_{tp_j}\Psi = \lr{-\dell_{p_j}p_t + i \dell_{p_j}q_t}\dot{q}_t
= \left[\dot{q}_t(-D_t + i B_t)\right]_j.\]
Starting from \eqref{E:q-deriv}, we have
\[\dell_{q_k q_j}\Psi = \dell_{q_j}p_t \dell_{q_k}q_t + i
\lr{\dell_{q_j}q_t\dell_{q_k}q_t + \delta_{j,k}} = \left[A_t^TC_t + i
\lr{A_t^TA_t + \Id}\right]_{k,j}.\]
Finally, from \eqref{E:q-deriv} we obtain 
\[\dell_{v_kq_j}\Psi = \frac{1}{2}\lr{(\dell_{q_j}p_t)_k + i\lr{-(\dell_{q_j}q_t)_k+\delta_{k,j}}}.
\]
The Hessian of $\Psi$ at $(q_c,p_c,v_c)$ is therefore
\[
\lr{\begin{array}{ccc}
  \frac{i}{2} \Id& \frac{1}{2}\lr{C_t + i (-A_t + \Id)}& \frac{1}{2}\lr{D_t + \Id - i B_t }\\
\frac{1}{2}\lr{C_t^T + i (-A_t + \Id)} & - A_tC_t + i (A_t^2 + \Id) & - C_t^T B_t + i A_tB_t\\
\frac{1}{2}\lr{D_t + \Id - i B_t^T } & - B_t^T C_t + i B_t^TA_t & B_t^T(-D_t + i B_t)
\end{array}},
\]
as claimed. We now calculate the determinant of the Hessian. Multiplying the first row on the left by $A_t$ and adding the result to the second row and then multiplying the first row by $B_t^T$ on the left and adding the result to the third row shows that the determinant of $\text{Hess}(\Psi)(q_c, p_c, v_c)$ is the same as the determinant of
\[
\lr{\begin{array}{ccc}
  \frac{i}{2} \Id& \frac{1}{2}\lr{C_t + i (-A_t + \Id)}& \frac{1}{2}\lr{\twiddle{D}_t - i B_t }\\
\frac{1}{2}\lr{C_t^T + i \twiddle{A}_t} & -i \twiddle{A}_t  & \twiddle{A}_t\\
\frac{1}{2}\lr{\twiddle{D}_t + i B_t^T } &  i B_t^T & B_t^T
\end{array}},
\]
where for any matrix $K,$ we write $\twiddle{K}=\Id +K.$ Next, multiplying the first column by $2$ and subtracting it from the second column, we find that the determinant of $\text{Hess}(\Psi)(q_c, p_c, v_c)$ is the same as the determinant of
\[
\lr{\begin{array}{ccc}
  \frac{i}{2} \Id& \frac{1}{2}\lr{C_t - i \twiddle{A}_t}& \frac{1}{2}\lr{\twiddle{D}_t- i B_t }\\
\frac{1}{2}\lr{C_t^T + i \twiddle{A}_t} & - C_t^T  & \twiddle{A}_t\\
\frac{1}{2}\lr{\twiddle{D}_t + i B_t^T } & -  \twiddle{D}_t & B_t^T
\end{array}}.
\]
Next, by dividing the first row and column by $\sqrt{2/i}$, we find that the determinant of $\text{Hess}(\Psi)(q_c, p_c, v_c)$ is the same as the determinant of
\[
\lr{\begin{array}{ccc}
   \Id& \frac{1}{\sqrt{2i}}\lr{C_t - i \twiddle{A}_t}& \frac{1}{\sqrt{2i}}\lr{\twiddle{D}_t- i B_t }\\
\frac{1}{\sqrt{2i}}\lr{C_t^T + i \twiddle{A}_t} & - C_t^T  & \twiddle{A}_t\\
\frac{1}{\sqrt{2i}}\lr{\twiddle{D}_t + i B_t^T } &  - \twiddle{D}_t & B_t^T
\end{array}}
\]
times $(i/2)^d$. Next, writing $\Omega = \twomat{0}{\Id}{-\Id}{0},$ note that
\[\twomat{-C_t^T}{\twiddle{A}_t}{-\twiddle{D}_t}{B_t^T}=\twiddle{M}_t^T\Omega.\]
Further, observe that
\[
\lr{C_t - i \twiddle{A}_t,\quad \twiddle{D}_t- i B_t } = (-i \Id\quad \Id)\twiddle{M}
\]
and that
\[
\lr{\begin{array}{c}
  C_t^T + i \twiddle{A}_t\\
\twiddle{D}_t + i B_t^T
\end{array}} 
=\twiddle{M}^T\lr{\begin{array}{c}  i \Id \\ \Id \end{array}} .
\]
Hence, the determinant of $\text{Hess}(\Psi)(q_c, p_c, v_c)$ is the same as $(i/2)^d$ times the determinant of
\[\lr{
  \begin{array}{cc}
    \Id& 0 \\
    0 &  \twiddle{M}^T
  \end{array}
}   \lr{
  \begin{array}{cc}
    \Id& \frac{1}{\sqrt{2i}}\lr{-i\Id~~~\Id} \twiddle{M}\\
\frac{1}{\sqrt{2i}}
    \lr{\begin{array}{c}
      i\Id \\ \Id
    \end{array}} & \Omega

  \end{array}}.
\]
Therefore, the determinant of $\text{Hess}(\Psi)(q_c, p_c, v_c)$ equals $(2/i)^d$ times $\det(\twiddle{M})$ times the determinant of
\[\lr{  \begin{array}{cc}
    \Id& \frac{1}{\sqrt{2i}}\lr{-i\Id~~~\Id} \twiddle{M}\\
\frac{1}{\sqrt{2i}}
    \lr{\begin{array}{c}
      i\Id \\ \Id
    \end{array}} & \Omega

  \end{array}}.
\]
Take the transpose of this matrix and using the Schur complement formula, we find that the determinant of this matrix is the same as the determinant of 
\[\Id - \frac{1}{2i}\left[\lr{i\Id~~~\Id} \Omega \twiddle{M}^T \lr{-i \Id~~~ \Id}^T\right]=\frac{1}{2i} \lr{-\Id~~~i\Id} \Omega M^T \lr{-i \Id~~~ \Id}^T,\]
where we have used that
\[\frac{1}{2i}\left[\lr{i\Id~~~\Id} \Omega \lr{-i \Id~~~ \Id}^T\right]= 2i\Id.\]
Finally, since 
\[\lr{-\Id~~~i\Id} \Omega M^T \lr{-i \Id~~~ \Id}^T = A_t + D_t +i (B_t^T - C_t^T)\]
and the determinant is invariant under transposing the matrix, we find that 
\begin{equation} \label{DETFORM} \det \lr{\text{Hess}(\Psi)(q_c, p_c, v_c)}= 2^{2d}\det\lr{\Id + M_t}\cdot \det(A_t + D_t +i (B_t - C_t)),\end{equation} 
as claimed.



\end{proof}
To complete the proof of Proposition \ref{T:Wig-Prop}, we note that when $(x, \xi)$ and $t=0$, we have $M_t = I$ and  obviously the modulus of  \eqref{DETFORM} is a positive constant. Hence the same is true for $(x, \xi) $ in a tubular neighborhood of some positive radius around $\Sigma_E$. Thus,  Proposition \ref{T:Wig-Prop} follows by stationary phase for positive complex phases \cite{HoI}.\hfill $\square$

\subsection{\label{dv} Remark: Integrating Out $dv$ First} 
We remark that a somewhat simpler analysis is possible for the proof of Proposition \ref{T:Wig-Prop} by first integrating in $dv$,  to get a reduced oscillatory integral in $dq dp$. This is natural since the $dv$ integral is essentially a Fourier transform of a Gaussian, slightly distorted by the non-constant amplitude. In particular the Hessian $d_v^2 \Psi$ at the critical point is non-degenerate. Eliminating $dv$ we obtain the reduced integral,

\begin{equation} \label{dqdp} \ucal_{\hbar}(t, x,\xi)= \frac{1}{(2\pi \hbar)^d}  \int_{\R^d} \int_{\R^d} e^{\frac{i}{\hbar}  \Psi_c(t, q, p; x, \xi) }  A(\hbar, q, p, t) d q dp\end{equation}
with reduced phase,
\begin{equation} \label{REDPHASE}  \begin{array}{l}  \Psi_c(t, x, \xi; q, p) :=S_t(q, p)+(p_t - p) \lr{x- \frac{q_t + q}{2}} - \langle q_t - q, \xi \rangle +i\lr{\abs{x -(\frac{q_t+ q}{2}) }^2}

 \end{array}\end{equation}
Here, we did not  use the equation $ \lr{x- \frac{q_t + q}{2}} =0$. The amplitude is that of the original $dqdpdv$ integral evaluated at $v = q_t - q$ and multiplied by the constant Hessian determinant in $v$. We need to further apply stationary phase for complex phase functions in $dqdp$ to obtain this as a critical point equation. If we do that, we obtain the critical point equation for the real part of \eqref{REDPHASE} 
  simplify to, $$\left\{ \begin{array}{l} x = \frac{q_t + q}{2}, \\ \\ 
   d_q S_t - \half ( p_t - p)  ( I  - d_q q_t) -  
\langle d_q q_t, \xi \rangle + \xi   = 0, \\ \\ 
  d_p S_t - \half  ( p_t - p) d_p q_t  -  \langle  d_p q_t , \xi \rangle = 0. \end{array} \right. , $$
and by Lemma \ref{DERIVS} reduce further to,
$$\left\{ \begin{array}{l} 
 \left(  \half ( p_t + p) - \xi\right) [- I + d_q q_t]  =0 , \\ \\ 
= \left( \half  ( p_t + p)  - \xi) \right)   d_p q_t  = 0, \end{array} \right. $$ 
which (again) imply $\half(p_t + p) = \xi$. At a dominant critical point, the real part of the Hessian is given by,
$$
\Re \rm{Hess}_{q, p} \Psi_{c, crit} = \begin{pmatrix}   & \\
A_tB_t -\half C_t (I + A_t)   & [ A_t D_t - \half B_t C_t]^T     \\ &&\\
A_t D_t - \half B_t C_t   & B_t D_t - \half D_t C_t  + \half C_t
\end{pmatrix} |_{\Theta^t(q, p) = (x, \xi)},
$$
and the imaginary part of the Hessian is given by, 
$$
\Im \rm{Hess} = \begin{pmatrix}  \frac{1}{4} [2 I 
+ 2 \langle \partial_q  q_t, \partial_q q_t \rangle + 2 \partial_q q_t] & 2 \langle \partial_p  q_t, \partial_q q_t \rangle   + 2 \partial_p q_t 
\\ & \\
2 \langle \partial_p  q_t, \partial_q q_t \rangle   + 2 \partial_p q_t  & \half  \langle \partial_p  q_t, \partial_p q_t \rangle
\end{pmatrix} 
$$
We leave this approach at this point since we have already calculated the full Hessian determinant.

\subsection{\label{S:Rem-Lag}Remark: Interpretation of Proof of Proposition \ref{T:Wig-Prop} by Lagrangian Submanifolds}
For fixed $(t, x, \xi)$, the  critical set of the phase \eqref{PSI}  in the sense of \cite{HoIII}  is therefore 

$$C_{W_{\hbar, t}} = \{(x, \xi, q, p, v): x = \frac{q+q_t}{2},\quad \xi = \frac{p+p_t}{2},\quad v= q_t-q\}. $$
Using Lemma \ref{DERIVS}, we find that the associated Lagrangian submanifold of $T^*(\R \times T^*\R^d)$ is given by,
\begin{equation} 
\label{Lambdat} \begin{array}{lll} \Lambda_{t} & = &  \{( x, d_x \Psi, \xi, d_{\xi} \Psi: (x, \xi, q, p, v) \in C_{W_{\hbar, t}(x, \xi))}\} \\ &&\\
& = & \{(x, (p_t-p), 
\xi,q_t - q),  x = \frac{q+q_t}{2},\quad \xi = \frac{p+p_t}{2}.\}\} \end{array} 
\end{equation}
We also consider the space-time Lagrangian \begin{equation} 
\label{Gamma} \begin{array}{lll} \Gamma & = &  \{t, d_t \Psi, x, d_x \Psi, \xi, d_{\xi} \Psi: (t,x, \xi, q, p, v) \in C_{W_{\hbar, t}(x, \xi))}\} \subset T^* (\R \times T^* \R^d)) \\ &&\\
& = & \{(t, - H(q, p), x, (p_t-p), 
\xi,q_t - q),  x = \frac{q+q_t}{2},\quad \xi = \frac{p+p_t}{2}.\} \end{array} 
\end{equation}

We have the natural projection,
\begin{equation} \label{pit} \pi : \Gamma \to \R \times T^*\R^d, \end{equation}
given by $$(t, - H(q, p) + \dot{p}_t(x + \frac{q_t - q}{2} - q_t), x, (p_t-p), 
\xi,q_t - q), (q, p, v)  \mapsto (t, x, \xi). $$
Since $v = q_t -q$  is uniquely determined  once $(q, p)$ are determined,  the fiber of \eqref{pit}  is the
set of solutions of $$\pi^{-1}(t, x, \xi) =  \{(q, p, v): x = \frac{q+q_t}{2},\quad \xi = \frac{p+p_t}{2},\quad v= q_t-q\} \leftrightarrow \{(t, \tau, q, p): \Theta^t(q, p) = (x, \xi)\}. $$
as defined in \eqref{PSItDEF}.  By the inverse function theorem, the solution is locally unique if $D_{(q, p)} \Theta^t$ does not have $-1$ as an eigenvalue. 
As an example where it is non-unique, we could let $(x, \xi) =(0,0)$ and consider `anti-podal' times when 
$\Theta^t(q, p) = - (q, p)$. For instance, in the case of the isotropic oscillator such a time is $t = \pi$. An oscillator 
on $\R^d$ with $k \leq d$ equal frequencies provides an example where the fiber is not discrete but has dimension
$d-k$.

\section{\label{SMOOTH-Proof}Proof of Theorem \ref{T:SMOOTH-Prop}: the $t$ Integral and Energy Asymptotics}

In this section we extend the analysis from the previous section, which involved computing the integral \eqref{E:Ucal-PSI} in $v,q,p$ to computing pointwise asymptotics for the Wigner transform $\ucal_\hbar(t,x,\xi)$ of the propagator, to include also the $t$ integral (see eg \eqref{Whfdelta})
\begin{equation}\label{E:Wigner-FIO-again}
W_{\hbar,f,1,E}(x,\xi)\sim \int_{\R^{3d}}e^{\ihbar (\Psi(q,p,v;t,x,\xi)+tE)}a_\hbar(t,q,p)\frac{dqdpdv}{(2\pi \hbar)^d}\frac{dt}{2\pi}   
\end{equation}
to determine the asymptotics of the Wigner distribution of the smoothed spectral projector.

As  in \cite{O98, TL}, it is the $t$-integral which introduces a degeneracy in the phase and a fold singularity in the associated Lagrangian submanifold (see Section \ref{FOLDBACK} and Section \ref{EFOLDSECT} for the geometric analysis of folds).  In the next Section \ref{MALGRANGE}, we put the phase into Malgrange normal form, 
or more correctly a stronger version  of it for cubic phases due to Chester-Friedman-Ursell \cite{CFU} and Levinson \cite{L61}. 

We consider two approaches to evaluating \eqref{E:Wigner-FIO-again}: (i) integrate in $(t, q, p, v)$ at once; (ii) integrate in $(q, p, v)$ first, reduce to a
one dimensional integral in $dt$ and then determine the asymptotics of this integral. The advantage of (i), which we carry out in \S \ref{S:tpqv-int}, is
that the critical point equations are considerably simpler to solve. However, the phase is complex-valued and the standard normal forms results need to be modified to apply to it.  The advantage of (ii), which we carry out in \S \ref{ONEDRED}, is that we can directly apply the asymptotics results from the literature \cite{CFU,HoI, GSt}. But the critical point equations become more complicated.

\subsubsection{\label{S:tpqv-int}The $dt dq dp dv$ integral} We extend Proposition \ref{PROPPROP} and Proposition \ref{P:det-Hess} by including the additional
integration in $t$.

\begin{proposition}\label{PROPPROPt}
Let $f\in \scal(\R)$. The Wigner function \[
W_{\hbar, E, f}(x,\xi) = \int_{\R^{3d}}  \int_{\R} \hat{\rho}(t) W_{\hbar, t}(x, \xi) e^{\frac{it E}{\hbar}} dq dp dv dt
\]
of the smoothed spectral function is a semi-classical Fourier integral kernel with a complex phase. The critical points in $ (t, q,p, v)$ at which the imaginary part of the phase 
$$\wt \Psi_E(t, q, p, v; x, \xi): = \Psi(t,q,p,v;x,\xi) + t E$$ 
vanishes  are all solutions of
\[
x = \frac{q+q_t}{2},\quad \xi = \frac{p+p_t}{2},\quad v= q_t-q,\quad H(q, p) = E.
\]

\end{proposition}

We postpone the proof of the last statement to Section \ref{HesstINT}, where we first integrate out the $(q, p, v)$ variables and reduce
to a one dimensional integral.

\begin{proof} The additional integration leads to the additional critical point equation,

 $$d_t \left( S(t, q, p)+ p_t (x + \frac{v}{2} - q_t) + t E +\frac{i}{2}\lr{\abs{x+\frac{v}{2}-q_t}^2+\abs{x-\frac{v}{2}-q}^2}) \right)= 0, $$
 whose real part is the equation,
 $$\begin{array}{l} d_t \left( S(t, q, p)+ p_t (x + \frac{v}{2} - q_t)  + t E\right)= 0\\ \\
  \iff \dot{q}(t) p(t) - H(q, p) + \dot{p}(x + \frac{v}{2} - q_t)
 - p_t \dot{q_t} + E = 0. \end{array}$$
 Due to the vanishing of the  imaginary part of the equation,  we also get $ \dot{p}(x + \frac{v}{2} - q_t)
 - p_t \dot{q_t} = 0$ at the critical point, and the additional critical point  equation simplifies to 
 $$\dot{q}(t) p(t) - H(q, p) 
 - p_t \dot{q_t}  + E= 0 \iff H(q, p) = E. $$
This proves all but the final statement, deferred to Section \ref{HesstINT}.

\end{proof}

\subsection{\label{ONEDRED} Reduction to a one-dimensional integral}
We now re-do the proof of Proposition \ref{PROPPROPt} by first eliminating the variables $(q, p, v)$ by stationary
phase and reducing to a single $dt$ integral. To state the result, we need to introduce some notation which will
be clarified in the course of the proof. Recall the Hamiltonian arc with $ \half (\gamma_{t, x, \xi}(0) + \gamma_{t, x, \xi}(t)) = (x, \xi) $ of  \eqref{gammaDEF},  $$\gamma_{t, x, \xi}(s) = \Phi^s (q(t, x, \xi), p(t, x, \xi)), \;\;\; s \in [0, t]$$ whose endpoints are given by, $$ \gamma_{t, x, \xi}(0) =  (q(t, x, \xi), p(t, x, \xi)),\;
\gamma_{t, x, \xi}(t) =  \Phi^t (q(t, x, \xi), p(t, x, \xi)),  $$
where $(q(t, x, \xi), p(t, x, \xi))$ are  local solutions  of  $$(q(t, x, \xi), p(t, x, \xi)) = (\Theta^{t})^{-1}(x, \xi) \iff \half (I + \Phi^t) (q(t, x, \xi), p(t, x, \xi)) = (x, \xi). $$
Note that $(\Theta^{t})^{-1}(x, \xi) $ is multi-valued but has well-defined branches near $\Sigma_E$  (see Lemma \ref{IFTLEM}).

We now rewrite  \eqref{PSIc} in an advantageous form, 
\begin{equation} \label{Psic*} 
\Psi_c(t, x, \xi) := \int_{\hat{\gamma}_{t, x, \xi}} (pdq ) - \int_{\gamma_{t, x, \xi}}  H ds = \int_{\beta_{t, x, \xi}} \omega - \int_{\gamma_{t, x, \xi}}  H ds , \end{equation}
where $\hat{\gamma}_{t, x, \xi}$ is the completed  contour  obtained by adding in the chordal link to obtain the  oriented  closed curve  connecting the endpoints of $\gamma_{t, x, \xi}$
 by the  chord
 $\alpha_{t, x, \xi} (s) = (1 -s) \gamma_{t, x, \xi}(0)  + s \gamma_{t, x, \xi}(t)$ (see \eqref{rhoform} ).  Note that
 $$\begin{array}{l}\int_{\alpha_{t, x, \xi}} p dq = \int_0^1 ( (1 - s) p_{t, x, \xi}(0) + s p_{t, x, \xi}(t) ) d   (1 - s) q_{t, x, \xi}(0) + s q_{t, x, \xi}(t)  \\ \\
 = \left( \int_0^1 ( (1 - s) p_{t, x, \xi}(0) + s p_{t, x, \xi}(t) ) ds \right) \cdot \left(  q_{t, x, \xi}(t) -  q_{t, x, \xi}(0) \right)\\ \\
 = \langle \xi,  q_{t, x, \xi}(t) -  q_{t, x, \xi}(0) \rangle. \end{array}
$$
  As mentioned below \eqref{rhoform}, $\hat{\gamma}_{t, x, \xi}$   bounds the two-dimensional surface
   $\beta_{t, x, \xi} $ consisting 
  of line segments joining $(x, \xi)$
to points of the Hamilton orbit, the integral \eqref{rhoform} equals the oriented area
  $ \int_{\beta_{t, x, \xi} } \omega $
  where  $\omega = d p \wedge dq$ is the standard symplectic form of $T^*\R^d$.

\begin{proposition}

\label{SMOOTHWEYLPROP}  The Wigner function $W_{\hbar, E, \rho}(x,\xi) = \int_{\R} \hat{\rho}(t) W_{\hbar, t}(x, \xi) e^{\frac{it E}{\hbar}} dt$ of the 
  smoothed spectral function is a  semi-classical Fourier integral kernel with real phase,
\begin{equation} \label{PSIc2} \Psi_E(t, x, \xi) :=  \int_{\beta_{t, x, \xi} } \omega  
- t H(\gamma_{t, x, \xi} (0))  +  t E, \end{equation}
and with amplitude 
$\left|  \det (1 + M_j(q_j, p_j,  t_j(E)) \right|^{-\half} $, where $q_j,p_j,t_j(E)$ are as in Proposition \ref{T:Wig-Prop} and Proposition \ref{T:SMOOTH-Prop2}.
For $(x, \xi)$ near $\Sigma_E$, the  critical point in $t$ of the phase occur at times $t$  such that there exist $(q,p) \in \Sigma_E$ for which
$(x, \xi) = \Theta^t(q, p)$. If $t_{\pm}(E, x, \xi)$ are the times $t$ such that $\Theta^t(q, p) = (x, \xi)$ has a solution, then
at a critical point, $\frac{\partial^2}{\partial t^2} \Psi_E(t) =  - d H(\gamma_{t, x, \xi}(0)) (J_{t, x, \xi}(0)) =  (\frac{dt}{dE})^{-1}$.

\end{proposition}

\begin{proof} The equation \eqref{PSIc2} is the same as \eqref{Psic*}.   We start with a general principle: Let  $f: \R_x^{n} \times \R_y^{m} \to \R$ be a smooth function, and consider the critical point  equation $d_{x, y} f(x, y) = 0$. Let $(a,b)$ be a critical point such that    the partial  Hessian $(\frac{\partial^2}{\partial y_j \partial y_k} f)(a, b)$ is non degenerate. Then by the  implicit function theorem,   there
exists an open set $U \subset \R^n$ containing $a$ and a unique differentiable function $g: U \to \R^m$  such that $g(a) = b$ and 
$d_y f (x, g(x)) = 0$ identically in $U$. Moreover, the critical point equations $d_x f(x, g(x)) = 0$ and $d_y f(x, g(x)) = 0$
are equivalent in $U$ to $d_{x, y} f(x, y) = 0$. Indeed, $   \frac{\partial }{\partial x} f(x, g(x))  = \frac{\partial f}{\partial x} (x, g(x)) +
\frac{\partial f}{\partial y} (x, g(x)) \frac{\partial g(x)}{\partial x}. $ The second term is zero in $U$ by definition of $g(x)$,
and then the first term must be zero when the left side is zero.

We apply this to the case where $n =1, m = 3d$ and $f(t, q, p, v)$ of the phase \eqref{PSI}. It follows
from Proposition \ref{PROPPROP} that the non-degneracy condition is satisfied as long as $D \Theta^t$ is non-singular, i.e. $D \Theta^t$ does not have $-1$ as an eigenvalue. We then get locally defined branches $g: \R \to \R^{3d}$, 
$g(t) = (q_t, p_t, v_t)$ such that $d_{q, p, v} \Psi(t, (q_t, p_t, v_t)) = 0$. The phase of the additional $dt$ integral
is $ \Psi(t, (q_t, p_t, v_t))  +  t E$. By the general principal, its critical points are the same as for those of
$d_{t, q, p, v} \Psi =0$ and by Proposition \ref{PROPPROPt}  these are given by the equations stated there.

\subsubsection{Calculation of $\frac{d}{dt} \Psi_E (t_c)$: Reynolds formula} A direct proof without using implicit functions can be obtained from the   Reynolds transport formula for the first term of \eqref{Psic*}, 
$$\frac{d}{dt}  \int_{\beta_{t, x ,\xi}} \omega =    \int_{\partial \beta_{t, x, \xi} }  \iota_{J_{t, x, \xi}} \omega = \int_{\hat{\gamma}}   \iota_{J_{t, x, \xi}} \omega  .$$
Here, $J_{t, x, \xi} $ is the variational (Jacobi)  vector field along $\partial \beta_{t, x, \xi} =  \hat{\gamma}_{t, x ,\xi} $ obtained by varying $\gamma_{t, x, \xi}$ and the chord $\alpha_{t, x, \xi}$. $J_{t, x, \xi} $  is determined by
$\gamma_{t, x, \xi}(0)$, since  $\Phi^s \gamma_{t,x, \xi}(0) = 
\gamma_{t, x, \xi}(s). $ Hence,
$$J_{t, x, \xi} (s) : = \frac{\partial}{\partial t} \Phi^s \gamma_{t,x, \xi}(0)  = D_{ \gamma_{t,x, \xi}(0) } \Phi^s \left( \frac{d}{dt}  \gamma_{t,x, \xi}(0) \right). $$  $J_{t, x, \xi}(s)$ is a Jacobi field along $\gamma_{t, x, \xi}(s)$, i.e. the variational vector field obtained by varying  a curve of of Hamilton orbits.  It is not a Jacobi field  on $\alpha_{t, x, \xi}$ since these are not extremals, but for brevity we continue to call it a Jacobi field.
We compute the integral by writing $\hat{\gamma}_{t, x, \xi}  = \gamma_{t, x, \xi} \cup \alpha_{t, x, \xi} $ and studying each integral separately. Next, recall that $\Xi_H$ is the Hamilton vector field of $H$. Since $\gamma_{t, x, \xi}'(s) = \Xi_H(\gamma_{t, x, \xi}(s))$, we find \begin{equation}
\label{TERM1}\begin{array}{l}  \int_{\gamma_{t, x, \xi} } \iota_{J_{t, x, \xi}} \omega = \int_0^t \omega(J_{t, x, \xi}(s), \Xi) ds 
\\ \\= \int_0^t   d H(J_{t, x, \xi}(s) ) =  \int_0^t   (\Phi^s)^* d H(J_{t, x, \xi}(0) )ds = t d H(J_{t, x, \xi}(0)). \end{array} \end{equation}

On the other hand,  since $\frac{d}{ds} \alpha_{t, x, \xi} (s) = \frac{d}{ds} ( (1 -s) \gamma_{t, x, \xi}(0)  + s \gamma_{t, x, \xi}(t) )= \gamma_{t, x, \xi}(t) -
\gamma_{t, x, \xi}(0), $ and since along $\alpha_{t, x, \xi}$, 
$J_{t, x, \xi} = (1-s) J_{t, x, \xi}(0) + s J_{t, x, \xi}(t), $ we have
 $$\begin{array}{lll} \int_{\alpha_{t, x, \xi} } \iota_{J_{t, x, \xi}} \omega = \int_0^1 \omega(J_{t, x, \xi} (s), \alpha_{t, x, \xi}'(s)) ds
 & =& \omega ( \half (J_{t, x, \xi}(0) +  J_{t, x, \xi}(t)),   \gamma_{t, x, \xi}(t) -
\gamma_{t, x, \xi}(0)) = 0, \end{array} $$
since 
\[
\half (J_{t, x, \xi}(0) +  J_{t, x, \xi}(t)) = \half \frac{d}{dt} (\gamma_{t, x, \xi}(0) + \gamma_{t, x, \xi}(t))
= \frac{d}{dt} (x, \xi)=0.
\]
That leaves the second term of \eqref{Psic*}, 
\begin{equation} \label{TERM2} \frac{d}{dt}  \int_{\gamma_{t, x, \xi}}  H ds  = \frac{d}{dt} t  H(\gamma_{t, x, \xi}(0)) = 
 H(\gamma_{t, x, \xi}(0))  + t d H (J_{t, x, \xi}(0)). \end{equation} 
 
 It follows from \eqref{TERM1} - \eqref{TERM2} that
$$\frac{\partial}{\partial t} \Psi_E = - H(\gamma_{t, x, \xi} (0)) + E,  $$
proving Proposition \ref{SMOOTHWEYLPROP}.

\subsubsection{\label{HesstINT} Calculation of $\Psi_E''(t_c)$. }

We differentiate 
$\frac{\partial}{\partial t} \Psi_E = - H(\gamma_{t, x, \xi} (0)) + E  $
once more to get
$$\frac{\partial^2}{\partial t^2} \Psi_E = - \frac{d}{dt} H(\gamma_{t, x, \xi} (0)) = - d H(\gamma_{t, x, \xi}(0)) (J_{t, x, \xi}(0)). $$

We claim that if $t(E, x, \xi)$ is implicitly one of the solutions of $\Theta^t(q, p) = (x, \xi)$ with $(q, p) = \gamma_{t, x, \xi}(0) \in \Sigma_E$ then \begin{equation} \label{dtE} \left(\frac{dt(E, x, \xi)}{dE} \right)^{-1}= - d H(\gamma_{t, x, \xi}(0)) (J_{t, x, \xi}(0)) |_{t = t(E,x, \xi)}. \end{equation}
This follows because $t_{\pm}(E, x, \xi)$ may be defined as the locally unique pair of solutions near $t=0$ of
the equation $H((\Theta^{t})^{-1} (x, \xi)) = E$. Thus, $$\frac{d}{dE} H((\Theta^{t})^{-1}(x, \xi)) = -
d H_{(\Theta^{t})^{-1} (x, \xi)} \frac{d}{dt} ((\Theta{t})^{-1}(x, \xi) \cdot  \frac{dt}{dE} =1. $$
Here, as above, $(\Theta^{t})^{-1} $ is multi-valued, so that $\Theta^{t})^{-1} (x, \xi)$ is a set but near $\Sigma_E$ it is double-valued and we mean that the equations hold for either branch. One such branch is given by
$(\Theta^{t})^{-1} (x, \xi) = \gamma_{t, x, \xi}(0)$ (the other being $g_{t, x, \xi}(t)$ and $ \frac{d}{dt} ((\Theta^{t})^{-1}(x, \xi) =
J_{t, x, \xi} (0)$ (resp. $J_{t, x, \xi}(t)$. This proves \eqref{dtE}. Regarding the amplitude,
after eliminating $(q, p, v)$ in  Proposition \ref{P:det-Hess}, the amplitude $$a_0=\lr{\det\lr{A_t + D_t + i(B_t-C_t)}}^{1/2}\exp(-itH(q,p))$$  gets multiplied by 
 $\left(\det\lr{\text{Hess}(\Psi)}(q_c,p_c,v_c) \right)^{-\half}$ where the Hessian is, 
\[(-1)^d\cdot \det(1+M_t(q_c, p_c))\det\lr{A_t(q_c,p_c) + D_t(q_c,p_c) +i (B_t(q_c, p_c)-C_t(q_c, p_c))}.\]
This cancels all factors of the amplitude except for  $\det(I + M_t)$. This concludes the proof of Proposition \ref{SMOOTHWEYLPROP}.

\end{proof}

\subsection{Background on folds \label{FOLDBACK}}

In the next Section \ref{EFOLDSECT} we discuss fold singularities of Lagrangian 
submanifolds associated to the Wigner spectral function. In preparation, we review the definitions and
properties of folding maps $f: X \to Y$  for general manifolds. References include \cite{G, HoIII}.

We follow the exposition in \cite[Appendix 1, Page 109]{G}. Let $f: X \to Y$ be a smooth map of n-dimensional
manifolds and let $S \subset X$ be a hypersurface (codimension one submanifold) of $X$. Let $d V_Y$ be a volume
form on $Y$.

Then $f: X \to Y$ is a folding map with $S$ as the fold locus if the following are satisfied,
\begin{itemize}
\item $S$ is set of critical points of $f$, i.e. $D_s f: T_s X \to T_{f(s)} Y$ fails to be surjective ; \bigskip

\item At every $s \in S, $ the kernel $\ker D_s f$ is one-dimensional  and is transversal to $T_s S$;\bigskip

\item $f^* dV_Y$ vanishes to first order at each $s \in S$, i.e. $\det D_s f$ vanishes to first order on $S$.

\end{itemize}

In this case, one can find coordinates $x_1, \dots, x_n$ on $X$ and coordinates $y_1, \dots, y_n$ on $Y$ so
that $S = \{x_n = 0\}$ and $f(S) = \{y_n = 0\}$ and so that $f^* y_j = x_j$, for $j \leq n-1$ and $f^*y_n = x_n^2$.
Associated to the folding map is a canonical involution $\sigma: V \to V$ in a neighborhood $V$ of $S$ in $X$, 
fixing $S$ and having the local form $\sigma (x_1, \dots, x_{n-1}, x_n) = (x_1 \dots, x_{n-1}, - x_n)$.

\subsection{\label{EFOLDSECT} Analysis of the fold singularity as a singularity of the projection of the Lagrangian}


The  time integral is the semi-classical  Fourier transform of the $dq dp dv$ integral, and the associated Lagrangian
submanifold generated by the phase is the classical analogue of the Fourier transform, $(t, \tau) \to (-\tau, t)$ applied
to \eqref{Lambdat} with $\tau = E$, namely 
\begin{equation} \label{LAGDEF2} \begin{array}{lll} \Lambda_{E} & = &  \{(x, \xi, d_{x, \xi} \Psi_E): d_{q,p, v, t} \wt \Psi_E= 0\}
\subset T^*(T^*\R^d))\\&&\\
& = &\{ (x, \xi,  p_t - p, q  - q_t): x = \frac{q+q_t}{2},\quad \xi = \frac{p+p_t}{2}, H(q, p) = E \} \\&&\\ & = &  \{(x, \xi, -2(x - q), 2(\xi - p)):  
x = \frac{q+q_t}{2},\quad \xi = \frac{p+p_t}{2}, H(q, p) = E\}.\end{array} \end{equation}

Note that the change of the Lagrangian submanifold $\Lambda_t \to \Lambda_E$ of \eqref{Lambdat} under the  $dt$ integral corresponds to the symplectic map  $(t, \tau) \to (- \tau, t)$ in \eqref{Gamma}.
The phase $\wt \Psi_E$ induces the Lagrangian immersion, 

\begin{equation} \label{iotaE} \begin{array}{lll} \iota_{\wt \Psi_E}:  (t, q, p, x, \xi) \in C_{\Psi_E} &  \to &  (\frac{q+q_t}{2},  \frac{p+p_t}{2},   q  - q_t, p_t - p) \\ &&\\
& = & (x, \xi, -2 (x -q), 2 (\xi - p)), \end{array} \end{equation}
which is closely related to the map \eqref{PSItDEF}.


The set of $\lambda \in \Lambda_{E,x, \xi}$ on which the natural projection
\begin{equation} \label{CAUSTIC} \pi: \Lambda_{E} \to T^*\R^d \end{equation} 
is a submersion is called the regular set and its complement where $D_{\lambda} \pi$ is singular is called the (Maslov) singular
cycle. The image of the Maslov singular cycle under $\pi$ is called the `caustic'. When $\Sigma_E$ is strictly convex,
the image is $B_E$ and there is an apparent fold over its boundary $\Sigma_E$. We prove this folding statement along with the closely related    Proposition \ref{FOLDPROPintro}.



\begin{prop}\label{FOLDPROP} 

\begin{itemize} \item (i) \; If $\Sigma_E$ is convex, the  map  \eqref{CAUSTIC}    $\pi: \Lambda_{E} \to B_E$ has a fold singularity along the zero section $ \Sigma_E \times \{0\} \subset T^*(T^* \R^d))$. 

\item (ii)   The map $\wt \Theta: \R \times \Sigma_E \to B_E$ with $\wt \Theta (t, q, p) = \Theta^t(q, p)$   \eqref{PSItDEF} fixes the diagonal when $t=0$ and the kernel
of its derivative is $2 \frac{\partial}{\partial t} - \Xi_H , $ i.e. $D \wt \Theta \left(2 \frac{\partial}{\partial t} - \Xi_H\right) |_{t=0}= 0$. \end{itemize}

 \end{prop}

  This is the folding result
 relevant to the asymptotics of the smoothed  Weyl-Wigner asymptotics of Theorem \ref{T:SMOOTH-Prop}.
 First we prove Proposition \ref{FOLDPROPintro}.

\begin{proof}  We first relate \eqref{CAUSTIC} and \eqref{PSItDEF} through the diagram,
 
$$ \begin{array}{llll} \iota_{\wt \Psi_E} : C_{\wt \Psi_E}  & \rightarrow & \Lambda_{E} \\ &&\\
\pi  \downarrow && \downarrow \pi \\&&\\ \Theta:  \R \times \Sigma_E & \rightarrow & B_E,
 \end{array}$$
 where in each case $\pi$ denotes the natural projection to the base. This shows that $C_{\Psi_E}$ is the fiber-product 
 $\Lambda_E \times_{B_E} (\R \times \Sigma_E)$ with respect to $B_E$, i.e. pairs of points in the product 
 space which get mapped to the same point in  $ B_E$.  This description guides the analysis of the fibers of
 the maps in the diagram. 
 
 It is evident that the projection $\pi: C_{\wt \Psi_E} \to (\R \times \Sigma_E)$ is a diffeomorphism, since 
 $(t, q, p; x, \xi) \in C_{\wt \Psi_E}$ occurs only when $(x, \xi) = \Theta^t(q, p)$, so $(x, \xi)$ is redundant. 
 From the formula in \eqref{iotaE} we see that $(q, p)$ are uniquely determined by the point
 $(x, \xi, - 2(x-q), 2 (\xi - p))  \in \Lambda_E$; since $t$ is defined by $q_t = 2 x - q, p_t = 2 \xi -p$ ,  the point $(t, q, p, x, \xi) \in C_{\wt \Psi_E}$ is uniquely
 determined. Hence, also $\iota_{\wt \Psi_E} $ is a diffeomorphism. It follows that the fold singularities of $\pi: \Lambda_E \to B_E$ and $\Psi: \R \times \Sigma_E \to B_E$ are equivalent.

The fiber of the projection \eqref{CAUSTIC} is,
$$\begin{array}{lll} \pi^{-1}(x, \xi) & = &  \{(x, \xi, q - q_t, p_t - p): H(q, p) = E, (x, \xi) = (\frac{q+q_t}{2},  \frac{p+p_t}{2})\}
\\ && \\
& = & \{((x, \xi),  q - q_t, p_t - p)),\; (q, p) \in (\Theta^t)^{-1}(x, \xi) \cap \Sigma_E)\}. 
\end{array}$$
It is evident that the inverse image 
of the boundary under \eqref{CAUSTIC} is the  zero section $0_{T^*(T^*\R^d)}$ over $\Sigma_E$. It  is the fold locus of the the map $\Theta: \R \times \Sigma_E$ at $t=0$ and therefore of the natural projection, 
$\pi: \Lambda_{E} \to T^*\R^d. $
It is the fixed point set of the involution $\sigma: (t, q, p) \to (-t, \Phi^t(q, p))$.
reversing the endpoints of the chord whose midpoint is $(x, \xi)$.

 To verify that $\Theta$ has a fold singularity along $\{0\} \times \Sigma_E$, i.e. satisfies the conditions for a fold in Section \ref{FOLDBACK} Consider the curves,
$$\alpha(s) = \alpha_{q, p}(s) : = (2s, \Phi^{-s}(q, p)) : \R \to \R \times \Sigma_E. $$
Then $\alpha(0) =  (0, q, p)$,  $\alpha'(0) = 2 \frac{\partial}{\partial t} - \Xi_H$ and (referring to \eqref{PSItDEF})  $$\wt \Theta (\alpha(s)) = \half(\Phi^{-s}(q,p) + \Phi^{2s} (\Phi^{-s}(q,p)) = \half(\Phi^{-s}(q,p) + \Phi^{s}(q,p)), $$
hence $$(D_{(0, q, p)} \wt \Theta) \alpha'(0)  = \frac{d}{ds} |_{s=0} \wt \Theta (\alpha(s)) = - \Xi_H(q,p) + \Xi_H(q, p) = 0. $$
 we see that $\wt \Theta$ has a fold along the 
points $\{0\} \times \Sigma_E$.

\end{proof}

 \begin{cor} \label{IFTLEM}  If $\Sigma_E$ is convex, there exists $\epsilon > 0$ so that the map $\Theta$ 
 is 2-1 on  $(-\epsilon, 0) \cup (0, \epsilon) \times \Sigma_E$, i.e.  if $|t| \leq \epsilon$, $(x, \xi) \in B_E$  and $d((x, \xi), \Sigma_E) < \epsilon$ then there is a unique Hamiltonian arc of time $t$ (up to orientation reversal) for which $(x, \xi)$ is the chordal midpoint.
 \end{cor}
 
The statement follows from the local description of fold singularities in Section \ref{FOLDBACK}. 
 As noted above, the canonical involution is  orientation reversal,  $\sigma (t, q, p) = (-t, \Theta^t(q, p))$, i.e.  $\Theta^t(q, p) = \half((q, p) + \Phi^t(q, p)) =  \Theta^{-t}(\Phi^{t}(q, p)) = \half(\Phi^t(q, p) + \Phi^{-t} \Phi^t(q, p))  $

On the other hand, if $\Sigma_E$ is convex, then there always exists at least one closed orbit of the Hamiltonian flow on $\Sigma_E$ if $H$ is convex \cite{Rab,W}. If $(q,p)$ is a periodic point of period $T$,then $\iota_{\Psi}(0, q, p) = \iota_{\Psi}(T, q, p)$. 
In this case, $\iota_{\Psi}^{-1}(x, \xi) = \{(k T, q p)\}$ has a discrete inverse image. But it is also possible to have an inverse image of
positive dimension, for instance if $V(x)= ||x||^2$ and $(x, \xi) = (0,0)$ for instance if periodic orbits come in positive dimensional 
families. Another such fiber occurs at a non-period of the isotropic oscillator:  at time $T = \pi$, $\Phi^T(q,p) = - (q, p)$ and so 
$\iota_{\Psi}^{-1}(0,0) = \{T\} \times \Sigma_E. $

\begin{rem} The phase function $\Psi_E$ of \eqref{PSIc2} is equivalent to the phase function $\wt \Psi_E$ in the sense of \cite{HoIII, GSt},
namely they both parametrize the same Lagrangian. \end{rem}
 
\section{\label{MALGRANGE}  Proof of Theorem \ref{T:SMOOTH-Prop}}
The proof of Theorem \ref{T:SMOOTH-Prop} is based on a special case of the Malgrange preparation theorem 
which was proved earlier  by  Chester-Friedmann-Ursell \cite{CFU}.
The Malgrange preparation theorem asserts that if $f(t, x)$ is a smooth function of $(t, x)  \in \R \times \R^n$, and if $k$ is the 
positive integer so that 
\begin{equation} \label{COND} f(0,0) = 0, \frac{\partial}{\partial t} f(0,0) =0, \;\; \cdots, \frac{\partial^{k-1}}{\partial t^{k-1}} f(0,0) = 0, 
 \frac{\partial^{k}}{\partial t^{k}} f(0,0) \not= 0,\end{equation}
then there exists a smooth function $c(t, x)$ which is non-vanishing near $(0,0)$ and smooth functions $a_j(x), j=1, \dots, k-1$
so that $$f(t, x) = c(t, x) \left(t^k + a_{k-1}(x) t^{k-1} + \cdots a_0(x) \right). $$
We only apply this theorem when $k=3$ and use a stronger form from \cite{CFU, L61} and also proved in detail in 
\cite[Theorem 7.5.13]{HoI} (see also \cite[Page 444]{GSt}). With no loss of generality assume that $\frac{\partial^3 f}{\partial t^3}(0,0) > 0$. Then, there exists a $C^{\infty}$ function  $T = T(t, x)$
near $(0,0)$ with $T(0,0) =0,  \frac{\partial T}{\partial t}(0,0)  > 0$ and $C^{\infty}$ functions $\rho, \mu$ near $(0,0)$ such that
\begin{equation} \label{CFUNF} f(t, x)=  \frac{T^3}{3} + a (x) T + \mu(x), 
\end{equation}
with $a(0)= 0, \mu(0) = f(0, 0)$.

The detailed statement of the result  is given in \cite[Theorem 7.7.19]{HoI}.  
\begin{prop} \label{HORMPROP} Let $f(t, x)$ be real valued on $\R^{n+1}$ satisfying \eqref{COND} with $k=3$. Then 
there there exist $C^{\infty}$ functions $u_{0, \nu}, u_{1, \nu}$ such that, in the notation above,

$$\begin{array}{lll} \int u(t, x) e^{i \tau f(t, x)} dt & \sim &  e^{i \tau b(x)} \tau^{-\frac{n}{2}} 
\left( \tau^{-\frac{1}{3}} \rm{Ai}(\tau^{\frac{2}{3}} a (x)) \sum_{\nu=0}^{\infty} u_{0, \nu}(x) \tau^{-\nu} \right)+ \\&&\\
& + &  e^{i \tau b(x)} \tau^{-\frac{n}{2}} \left( \tau^{-\frac{2}{3}} \rm{Ai}'(\tau^{\frac{2}{3}} a (x)) \sum_{\nu=0}^{\infty} u_{1, \nu}(x) \tau^{-\nu} \right). \end{array} $$

\end{prop}

\begin{rem} Note that $T^2 + a(x) = 0$, so that $a(x) < 0$. In \cite[Page 444]{GSt}), the normal form is stated in the form,
$\frac{T^3}{3} - \rho(x) T + \mu(x). $ Hence, in what follows, $a(x) = - \rho(x)$.
\end{rem}

We apply the result to the oscillatory integral $W_{\hbar, E, \rho}(x,\xi) = \int_{\R^{3d}}  \int_{\R} \hat{\rho}(t) W_{\hbar, t}(x, \xi) e^{\frac{it E}{\hbar}} dq dp dv dt$ and to the one-dimensional reduction in Section \ref{ONEDRED}. 

 For the 
full $\R^{3d +1}$ dimensional integral, we use 






 \begin{prop} \label{HORMPROP2} Let $f(y, x)$ be complex  valued function with positive
imaginary part  on $\R^{n+m}$. Assume that  $\Im f (0,0)=0$, that $\rm{Hess} (\Im f)(0,0)$ is non-degenerate,  and that
\begin{itemize}
\item (i) \; $\Re f_x'(0,0) = 0. $

\item  (ii)\;  $\rm{Rank} \Re f_{xx}''(0,0) = n-1. $ 

\item (iii)\; $\langle X, \partial/\partial x \rangle^3 \Re f(0,0) \not= 0 $ if $0 \not= X \in \ker \Re f_{xx}''(0,0) $.

\item (iv) $f(x_1, 0)$ is real valued. 

\end{itemize}

Then there exist $C^{\infty}$ functions $\rho(x), \mu(x)$
near $(0,0)$ such that $\rho(0) = 0, \mu(0) = f(0)$ and such that, for any $u \in C_0^{\infty} $ with $\rm{Supp} u$ sufficiently close to $(0,0)$,
there exist $C^{\infty}$ fuctions $u_{0, \nu}, u_{1, \nu}$ such that,
$$\begin{array}{lll} \int u(y, x) e^{i \tau f(y, x)} dx  & \sim &  e^{i \tau \mu (y)} \tau^{-\frac{n}{2}} 
\left( \tau^{-\frac{1}{3}} \rm{Ai}(-\tau^{\frac{2}{3}} \rho(y)) \sum_{\nu=0}^{\infty} u_{0, \nu}(y) \tau^{-\nu} \right)+ \\&&\\
& + &  e^{i \tau \mu (y)} \tau^{-\frac{n}{2}} \left( \tau^{-\frac{2}{3}} \rm{Ai}'(- \tau^{\frac{2}{3}} \rho(y)) \sum_{\nu=0}^{\infty} u_{1, \nu}(y) \tau^{-\nu} \right). \end{array} $$

\end{prop}

The proof is similar to that of 
 \cite[Theorem 7.7.18]{HoI}  and also in the reduction to one integral. One integrates first in  $(x_2, \dots, x_n)$ where stationary phase applies, then applies Proposition \ref{HORMPROP} to the remaining one dimensional integral.  For the sake of completeness, we provide more details as 
in \cite[Theorem 7.7. 19]{HoI}, although it explicitly assumed there that the phase is real-valued. We label the coordinates so that $x_1 = t$ and $(q, p, v) = (x_2, \dots, x_n)$ 
with $n = 3d +1$. Let $f = \Re \Psi$. Define $X_j(x_1, y) $ for $j=2. \dots, n$ so that the equations $\frac{\partial f(x, y)}{\partial x_j}  = 0$ 
determine $x_j $ as functions $x_j = X_j(x_1, y)$. Change variables in the $dx_2 \cdots dx_n$ integral to
$x_j' = x_j - X_j$. Then the critical points of the $dx_2 \cdots d x_n$ integral become  $x_j = 0, j=2, \dots, n$. By applying stationary phase with a positive  complex phase to this integral, we obtain an asymptotic of the integrand 
with fixed $x_1$ of the form,
$$\tau^{- \frac{n-1}{2}} e^{i \tau f(x_1, 0)} U(x_1, y, 0). $$
Here we use that the imaginary part $\Im \Psi$  of the phase vanishes at the critical set.

To apply the results to our integrals, it is necessary to show that the hypotheses of Proposition \ref{HORMPROP} hold and then to  calculate the coefficients $\rho, \mu$ in \eqref{CFUNF} and the
principal coefficients $u_{0,0}, u_{1,0}$ in Proposition \ref{HORMPROP}.
The following general calculation is due to \cite[Section 3]{CFU}; see also \cite[(6.7), Page 440]{GSt}.
By \eqref{CFUNF}, and the convention that $\rho= -a$,  we may express  $\phi(t,x)  = \frac{T^3}{3} - \rho T + \mu. $ For each critical point
$d_t \phi(t,x)$ there exist two roots $t_{\pm}(x)$.  By $ \phi_c(t_{\pm},x)$ we denote the critical value of the
phase at each critical point. We then have,
\begin{lem} \label{rhoLEM}  In the notation above, $$\left\{ \begin{array}{l} \rho^{3/2}(x) =\frac{3}{4} ( \phi_c(t_-, x) - \phi_c(t_+, x) )\\ \\
\mu(x) = \frac{1}{2} ( \phi_c(t_-, x) + \phi_c(t_+, x) )\ . \end{array} \right. $$
\end{lem}

We review the proof, since we need to do the calculations.

\begin{proof} We assume that  $\phi(t,x) = \frac{T^3}{3} - \rho(x) T + \mu(x) $ is a smooth expression (here we write $a = - \rho$). The critical point equation in  $t$ 
(viewing $x$ as a parameter) is 
  $$\frac{\partial}{\partial t} \phi(t,x) = \frac{\partial }{\partial T} \phi(T,x) \frac{\partial T}{\partial t} =  \frac{\partial T}{\partial t}\left(T^2 - \rho(x) \right)= 0. $$
  Hence, since $\frac{\partial T}{\partial t}\neq 0$ for $t,x$ sufficiently close to $0$, we have $T^2 = \rho(x)$  on the critical point set, and there exist two roots  $T_{\pm} = \pm \sqrt{\rho(x)}$ when $T \not=0$. It follows that $\frac{\partial t}{\partial T} \not= 0$ when 
  $T \not= 0$. 
The critical values of the phase at the two critical points are given by, 
$$\phi_c(t_{\pm},x) = \frac{T_{\pm}^3}{3} - \rho(x) T_{\pm} + \mu(x) = \frac{\pm \rho^{\frac{3}{2}}}{3} - \rho(x) (\pm \rho^{\half})  + \mu(x)
= \mp \frac{2}{3} \rho^{\frac{3}{2}} + \mu(x).$$
 Thus,
  $$\phi_c(t_+, x) = - \frac{2}{3} \rho^{\frac{3}{2}} (x) + \mu(x),\;\; \phi_c(t_-, x) =  + \frac{2}{3} \rho^{\frac{3}{2}} (x) + \mu(x). $$


\end{proof}

\subsection{Completion of the proof of Theorem  \ref{T:SMOOTH-Prop}}

\subsubsection{Calculation of $\Psi_E''(t_c)$ and of $\Psi_E'''(t_c)$. }

\begin{proof}
When $t=0$ and $(x, \xi) \in \Sigma_E$, we have
$$\frac{\partial^2}{\partial t^2} \Psi_E = - \frac{d}{dt} H(\gamma_{t, x, \xi} (0)) = - d H(\gamma_{t, x, \xi}(0)) (J_{t, x, \xi}(0)) |_{t =0, (x, \xi) \in \Sigma_E}  $$
 In computing the time derivatives we may assume $(x, \xi) \in \Sigma_E$. But $ H(\gamma_{t, x, \xi} (0))$ has a local minimum when $t=0$ in that case. Indeed, for any $t$ such that   $(x, \xi) = \half (\gamma_{t, x, \xi}(0) + \gamma_{t, x, \xi}(t))$, it is necessary that  $H(\gamma_{t, x, \xi}(0)) \geq H(x, \xi) = E. $  It follows also that 
$\frac{d}{dt} (\Theta^t)^{-1}(x, \xi)|_{t=0} \in T \Sigma_E$, and that
\begin{equation} \label{LM}d H{(\gamma_{t, x, \xi} (0))} |_{t=0} = 0, \;\; (x, \xi) \in \Sigma_E. \end{equation}
We now show that  
$$\frac{\partial^3}{\partial t^3} \Psi_E |_{t=0} = - \frac{d^2}{dt^2} H(\gamma_{t, x, \xi} (0))  |_{t=0} \not= 0. $$
We have $\Theta^t(\gamma_{t, x, \xi}(0))
= (x, \xi)$ and $\frac{d}{dt} \Theta^t = \half \frac{d}{dt} \Phi^t = \half \Xi_H$. Then, $0 = \frac{d}{dt} \Theta^t(\gamma_{t, x, \xi}(0)) = \half \Xi_H(x, \xi)
+ D \Theta^t \frac{d}{dt} \gamma_{t, x, \xi} |_{t =0}$. If  the second term is  zero, we would get $\Xi_H(x, \xi) = 0$, a contradiction. In fact,
$D \Theta^t \frac{d}{dt}|_{t =0} = Id$ and we get $J_{0, x, \xi}(0) = - \half \Xi_H(x, \xi)$. But then, by strict convexity of $H$,
$$ \frac{d^2}{dt^2} H(\gamma_{t, x, \xi} (0))  |_{t=0 }= \rm{Hess}_{x, \xi} H (J_{0, x, \xi}(0) , J_{0, x, \xi}(0) ) > 0.$$
Note that the other term $d H_{(\gamma_{t, x, \xi} (0))} \frac{d^2}{dt^2} \gamma_{t, x, \xi}(0)  = 0$ by \eqref{LM}. This completes the proof that the conditions \eqref{COND} are satisfied.

We may therefore apply Proposition \ref{HORMPROP} to obtain an asymptotic expansion of the form stated in 
Theorem \ref{T:SMOOTH-Prop}. The remaining step is to use Lemma \ref{rhoLEM} to calculate the parameters
$\rho, \mu, u_{0,0}, u_{1, 0}$.

We therefore fix $(x, \xi) \in B_E$ sufficiently close to $\Sigma_E$ and define $(t, q, p) \in (-\epsilon_0, \epsilon_0)  \times \Sigma_E$ as above by $\Psi^{t}(q, p) = (x, \xi)$. Further, we denote by $t_{\pm}(E, x, \xi)$ the two roots
of the critical point equation for the phase $\Psi_E(t; x, \xi)$ (see Definition \ref{tDEF}).

\begin{lem} \label{ODDLEM} If $(x, \xi)$ is $\epsilon_0$-close to $\Sigma_E$, then there exist exactly two critical times $t = t_{\pm}(E, x, \xi) \in (-\epsilon_0, \epsilon_0)$ (see Proposition \ref{PROPPROP})  and they are of the form $t_{\pm} = \pm t$ where $\Theta^t(q, p) = (x, \xi), H(q, p) = E. $ Moreover, the action changes sign under the reversal of $\gamma_{t, x, \xi}$.  \end{lem}

Indeed, for such $(x, \xi)$ there is a unique Hamiltonian arc of time $t$ with $|t| \leq \epsilon$  (up to the symmetry \eqref{SYM}), for which $(x, \xi)$ is the midpoint of the chord with  endpoints $(q, p), (q_t, p_t)$. The two critical times points $(t, q, p)$ correspond to $\gamma_{t, x, \xi}$ and its orientation reversal.   Reversing the orientation means 
to change  $\gamma_{t, x, \xi}(s) \to \gamma^*_{t, x, \xi} (s): = \gamma_{t, x, \xi}(t -s)$, so that $\frac{d}{ds}  \gamma^*_{t, x, \xi} (s):  = - (\frac{d}{ds}  \gamma_{t, x, \xi}) (t-s) $.
We write $ \gamma_{t, x, \xi} (s) = (q_{t, x, \xi}(s), p_{t, x, \xi}(s))$ and $ \gamma^*_{t, x, \xi} (s)=
(q^*_{t, x, \xi}(s), p^*_{t, x, \xi}(s))  = (q_{t, x, \xi}(t-s), p_{t, x, \xi}(t-s))  $ and show that 
$$\int_0^t  p_{t, x, \xi}(s)) \frac{d}{ds}  q_{t, x, \xi}(s)) ds = - 
\int_0^t p^*_{t, x, \xi}(s)) \frac{d}{ds}  q^*_{t, x, \xi}(s)) ds. $$
Indeed, the right side equals 
$$\begin{array}{l} \int_0^t p_{t, x, \xi}(t- s)) \frac{d}{ds}  q_{t, x, \xi}(t-s)) ds  = - \int_0^t p_{t, x, \xi}(t- s))  \dot{q}_{t, x, \xi}(t-s)) ds\\ \\ =  \int_t^0 p_{t, x, \xi}(s))  \dot{q}_{t, x, \xi}(s)) ds
= -  \int_0^t p_{t, x, \xi}(s))  \dot{q}_{t, x, \xi}(s)) ds. \end{array}  $$

\subsubsection{Calculation of $\rho$} 
By the first equation of Lemma \ref{rhoLEM},  we need to calculate  the ``odd'' part of the  phase,
defined by $$\begin{array}{l} \half (\Psi_E(t_{+}(E, x, \xi); x, \xi) - \Psi_E(t_{-}(E, x, \xi); x, \xi)) \\ \\= \half (S_{t_{+}(E, x, \xi)} - S_{t_{-}(E, x, \xi)}) - \langle \half (q_{t_{+}(E, x, \xi)} - q_{t_{-}(E, x, \xi)}), \xi \rangle + (t_{+}(E, x, \xi) - t_{-}(E, x, \xi)) E, \end{array}$$
where $S_t(q, p) = \int_{\gamma_{t, q, p}} (pdq - H ds)$ is the action integral along the phase space trajectory
$\gamma_{t, q, p}$ with initial value $(q, p)$. The`` odd"  part of the action integral is 
$$\half (\int_0^{t_{+}(E, x, \xi)}  p dq - \int_0^{t_{-}(E, x, \xi)} p dq) = \half \int_{t_{-}(E, x, \xi)}^{t_{+}(E, x, \xi)} pdq, $$
 where the integral is taken along the Hamilton path with endpoints $(q_{t_-}, p_{t_-}), (q_{t_+}, p_{t_+})$
 where $t_{\pm} = t_{\pm}(E, x, \xi)$. 
Since  $H(q_s, p_s) = H(q, p)$, the second term of $S_t$  combines with 
$t E$ to produce $(t_{+}(E, x, \xi) - t_{-}(E, x, \xi))  (E - H(q, p)) = 0$ at the critical time. One also has, 
 $$\langle \half (q_{t_{+}(E, x, \xi)} - q_{t_{-}(E, x, \xi)}), \xi \rangle = 
 \half \int_{t_{-}(E, x, \xi)}^{t_{+}(E, x, \xi)} \xi \cdot dq. $$
  Hence,
 \begin{equation} \label{rho32} \begin{array}{l} \frac{4}{3} \rho^{3/2}=  \half \int_{t_{-}(E, x, \xi)}^{t_{+}(E, x, \xi)} (p - \xi) \cdot dq =    \int_{\beta_{t, x, \xi} } \omega, \end{array} \end{equation} where in the second inequality we use the observation 
  noted in \eqref{rhoform} and in the proof of Proposition \ref{SMOOTHWEYLPROP} that
 this action integral is the oriented area  of the surface bounded by the oriented closed curve $\hat{\gamma}_{t, x, \xi}$ 
 consisting of the Hamilton arc followed by the chord.
 
 

 \subsubsection{Calculation that $\mu = 0$}

We claim that $\mu =0$. This follows from Lemma \ref{ODDLEM}, since  the two critical times correspond to the 
unique Hamilton orbit and  its orientation reversal. The same is  true of $\langle q_t - q, \xi \rangle$, which is clearly 
odd under orientation reversal.

\subsubsection{\label{u00SECT} Calculation of $u_{0,0}$}
We claim that
\begin{equation}\label{MATCH}  u_{0,0} = \sqrt{\pi} \rho^{1/4} \left| \frac{dt_j}{dE} \det (1 + M_j(x, t_j(E))\right|^{-\half},  \end{equation}

This follows from  a standard matching result,  using the stationary phase asymptotics of the Airy function, and is given in detail in \cite[Page 459]{GSt}. 
We wish to match the  Airy asymptotic expansion \cite[Page 442, (6.9)]{GSt}, 
$$e^{i \tau \sigma} \{\frac{g_0}{\tau^{1/3}} \rm{Ai}(- \tau^{2/3} \rho) +  \frac{g_1}{i \tau^{2/3}} \rm{Ai}'(-\tau^{2/3} \rho) \}$$ to the  stationary phase of the same integral for  for $\tau^{2/3} \rho >> 0$, given by 
 $$\frac{\tau^{-1/3} e^{i \tau \sigma}}{\sqrt{\pi} [\tau^{2/3} \rho]^{1/4}} \left(g_0 \cos (\frac{2 \tau \rho^{2/3}}{3} - \pi/4)) - 
 g_1 \rho^{1/2} \sin (\frac{2 \tau \rho^{2/3}}{3} - \pi/4) \right). $$
 See  \cite[Page 443, (6.15)]{GSt}.  The coefficients of the expansions match if and only if (in an obvious notation)  $$\rm{Airy} \;\; g_0 \iff \rm{StPh} \;\; \frac{1}{\sqrt{\pi}} \rho^{-1/4} g_0. $$
 
Note that  $\tau^{-1/3}[ \tau^{-2/3} ]^{1/4} = \tau^{-1/3 - 1/6} = \tau^{-\half}.$ The stationary phase calculation uses that,
 $$\rm{Ai}(t) \simeq \frac{1}{\sqrt{\pi} t^{1/4}} \cos (\frac{2}{3} t^{3/2} - \frac{\pi}{4}). $$


For stationary phase points very close to $\Sigma_E$ one gets by Proposition \ref{T:SMOOTH-Prop2} the amplitude 
$$ \frac{2^{d+1}}{\sqrt{2 \pi \hbar}} \sum_j e^{- \epsilon t_j/\hbar} A_j(x, E) \cos \left( \frac{S_j(x, E)}{\hbar} + m_j\right), $$
and applying this to the critical points $(t, q, p)$ corresponding to $(x, \xi)$ close to $\Sigma_E$ gives that 
$$g_0 = \left| \frac{dt_j}{dE} \det (1 +M_j(x, t_j(E))^{-1} \right|^{\half},$$
proving \eqref{MATCH} and 
concluding the proof of Theorem  \ref{T:SMOOTH-Prop}.

\end{proof}


  \section{\label{ISOHO} Example: Isotropic harmonic oscillator. }
  In this section, we evaluate all of the objects above in the simplest case of the isotropic harmonic oscillator,
 $H(q, p) = \half (||p||^2 + ||q||^2)$ and check the consistency of \eqref{WAIRY} and \eqref{AIRY}.

  The first simplifying feature is that the Hamilton flow $\Phi^t$, resp. the midpoint map $\Theta^t = \half (I + \Phi^t)$ are linear; they are
  given respectively by,
\begin{equation} \label{Phit} M_t = \Phi^t = \bigoplus  \begin{pmatrix} \cos  t &  \sin  t \\ &\\
- \sin t  &  \cos t \end{pmatrix}, \;\; \Theta^{-t} = 
\bigoplus \frac{2}{1 + \cos  t }    \begin{pmatrix}  (1 + \cos   t ) & -   \sin t \\ &\\
 \sin t  & (1 + \cos t ),\end{pmatrix}. \end{equation} 
 i.e. 
 $(q_t, p_t) = ( \cos t q + \sin t p, - \sin t q + \cos t p)$ and on the critical point set of the $(q, p, v)$ integral,
 \begin{equation} \label{pqHO}  \begin{pmatrix}q \\ \\ p \end{pmatrix} = \Theta^{-t} (x, \xi) = \frac{2}{1 + \cos t} \begin{pmatrix} (1 + \cos t) x - \sin t \xi \\ \\
 \sin t x + (1 + \cos  t) \xi \end{pmatrix}.\end{equation}

Given $(x, \xi)$ we can solve for the critical time $t = t(x, \xi, E)$ at which there exists $(q,p) \in \Sigma_E$
for which $(x, \xi) $ is the midpoint of the Hamilton orbit of time $t$ starting at $(q,p)$. Namely, 
$$\cos \frac{t}{2} = \frac{ \sqrt{H(x, \xi)} }{ \sqrt{H(q, p)}} = \sqrt{s}, \;\; (H(q, p) = E)$$
where in the second equality we use  the notation $ s =  \frac{H(x, \xi)}{E} \in (0,1]$ in \eqref{BFORMULA}. 
 There are two solutions $t = t_{\pm}(E, x, \xi)$ with $t_- = - t_+$.
 
 \subsubsection{Simple geometric formulae for spheres} Let $C_R$ be the circle of radius $R$ around the origin. Since the Hamiltonian orbits of an isotropic oscillator are  great circles  of the sphere
of radius $\sqrt{2 E}$, we may regard them as circular arcs on $C_{\sqrt{2 E}}$. The travel time $t$  of the Hamiltonian flow $\Phi^t$ on this arc 
has length $R t = \sqrt{2 E} \;t$.

The chord between endpoints of the Hamilton arc of time $t$ has length $\half \sqrt{R^2 - r^2}$ where $r$ is the distance from the origin
to the midpoint of the chord. Since $(x, \xi)$ is the midpoint of the chord, $r = \sqrt{2 H(x, \xi)}$ and the length of the chord is
$\half \sqrt{2 E - 2 H(x, \xi)}$.  The area of the circular sector bounded by the Hamilton arc and the chord is given by,
$$A = R^2 \cos^{-1} (\frac{r}{R}) - r R \sqrt{1- \frac{r^2}{R^2}} = (2 E) \cos^{-1} \sqrt{\frac{H(x, \xi)}{E}} - 2 \sqrt{H(x, \xi) E} \sqrt{1 - \frac{H(x, \xi)}{E}}. $$
In the 
In the notation of  \eqref{BFORMULA} with $\sqrt{s} = \frac{r}{R}$, and comparing with \eqref{rho32},
$$A = R^2 \left( \cos^{-1} \sqrt{s} - \sqrt{s} \sqrt{1 - s} \right) = \frac{4}{3} \rho^{3/2}. $$


On the other hand,  $\beta(s) := \half [\cos^{-1} \sqrt{s} - \sqrt{s - s^2}],  $ so $A = 2 R^2 \beta$ and 
$B^2 = -[ (2 R^2)^{-1} A ]^{\frac{2}{3}}$. We recall that $\hbar = \hbar_N(E):=\frac{E}{N+\frac{d}{2}}$ and that the argument of 
the Airy function in \eqref{WAIRY} is  $\hbar_E^{-2/3}B^2\lr{s}$. We cannot and do not use the same convention for a general Hamiltonian, so that the factor $[(2 R^2)^{-1} ]^\frac{2}{3}  = E^{-\frac{2}{3}}$ accounts for the change in the  Airy argument between $\hbar_E^{-2/3}B^2\lr{s}$ in  \eqref{WAIRY} and $\hbar^{-\frac{2}{3}} \rho(x, \xi))$
\eqref{AIRY}.

\subsection{Simplified asymptotics}

For $(x, \xi)\in T^*\R^d$ define the rescaled variable $u=u(x,\xi)$ centered at the energy surface $\Sigma_E$ by
\begin{equation} \label{TUBE} \lr{\norm{x}^2+\norm{\xi}^2}/2 =E+ u  \lr{\hbar/2E}^{2/3}.\end{equation} 
 It is proved in \cite{HZ20} that,  for $\abs{u}<\hbar^{-1/3}$, 
\begin{equation}
  \label{E:W-scaling}
 W_{\hbar, E}(x, \xi)=
  \begin{cases}
 \frac{2}{(2\pi \hbar)^d} \lr{\frac{\hbar}{2E}}^{1/3} \lr{\Ai(u/E) + O\lr{(1+\abs{u})^{1/4}u^2\hbar^{2/3}}},&\qquad u<0\\
 \frac{2}{(2\pi \hbar)^d} \lr{\frac{\hbar}{2E}}^{1/3} \Ai(u/E)\lr{1 + O\lr{(1+\abs{u})^{3/2}u\hbar^{2/3}}},&\qquad u>0
  \end{cases}
\end{equation}
The simplified asymptotics follow by Taylor expansion of $\beta$ in   \eqref{BFORMULA} and of the
amplitude $u_{0,0}$ denoted in \cite{HZ20} by 
\begin{equation}
 \label{E:FW_alpha0} \alpha_0(s) =
 \begin{cases}
s^{(1-\alpha)/2} \frac{\sqrt{2|B(s)|}}{(1-s)^{1/4}s^{3/4}},&\quad 0<s<1\\
s^{(1-\alpha)/2} \frac{\sqrt{2B(s)}}{(s-1)^{1/4}s^{3/4}},&\quad s\geq 1
 \end{cases}.
\end{equation}
Namely, the Taylor expansions are,
 \begin{equation}
B^2(1+t)=2^{-2/3}t(1+ O(t)),\qquad \alpha_0(1+t)=2^{1/3}(1+O(t)).\label{E:alpha-B-exp}
\end{equation}
These estimates yield
\begin{equation}
\nu^{2/3}B^2\lr{s}= \lr{\frac{2}{\hbar_E}}^{2/3}B^2(1+(u/E)\hbar_E^{2/3})= \frac{u}{E} + O\lr{\abs{u}^2\hbar_E^{2/3}},\label{E:B-exp}
\end{equation}
which leads in \cite{HZ20} to,
\[\nu^{-1/3}\Ai(\nu^{2/3}B^2(s))\alpha_0(s)=\hbar_E^{1/3}\left[\Ai(u/E)+O((1+\abs{u})^{1/4}u^2\hbar_E^{2/3})\right]\]
when $u<0$ and 
\[\nu^{-1/3}\Ai(\nu^{2/3}B^2(s))\alpha_0(s)=\hbar_E^{1/3}\Ai(u/E)\left[1+O((1+\abs{u})^{3/2}u\hbar_E^{2/3})\right]\]
when $u>0.$

\end{document}